\PassOptionsToPackage{unicode}{hyperref}
\PassOptionsToPackage{hyphens}{url}
\PassOptionsToPackage{dvipsnames,svgnames,x11names}{xcolor}

\documentclass[
  12pt]{article}

\usepackage{amsmath,amssymb}
\usepackage{iftex}
\usepackage[T1]{fontenc}
\usepackage[utf8]{inputenc}
\usepackage{textcomp} % provide euro and other symbols

\usepackage{lmodern}
\IfFileExists{upquote.sty}{\usepackage{upquote}}{}
\IfFileExists{microtype.sty}{% use microtype if available
  \usepackage[]{microtype}
  \UseMicrotypeSet[protrusion]{basicmath} % disable protrusion for tt fonts
}{}
\makeatletter
\@ifundefined{KOMAClassName}{% if non-KOMA class
  \IfFileExists{parskip.sty}{%
    \usepackage{parskip}
  }{% else
    \setlength{\parindent}{0pt}
    \setlength{\parskip}{6pt plus 2pt minus 1pt}}
}{% if KOMA class
  \KOMAoptions{parskip=half}}
\makeatother
\usepackage{xcolor}
\makeatletter
\ifx\paragraph\undefined\else
  \let\oldparagraph\paragraph
  \renewcommand{\paragraph}{
    \@ifstar
      \xxxParagraphStar
      \xxxParagraphNoStar
  }
  \newcommand{\xxxParagraphStar}[1]{\oldparagraph*{#1}\mbox{}}
  \newcommand{\xxxParagraphNoStar}[1]{\oldparagraph{#1}\mbox{}}
\fi
\ifx\subparagraph\undefined\else
  \let\oldsubparagraph\subparagraph
  \renewcommand{\subparagraph}{
    \@ifstar
      \xxxSubParagraphStar
      \xxxSubParagraphNoStar
  }
  \newcommand{\xxxSubParagraphStar}[1]{\oldsubparagraph*{#1}\mbox{}}
  \newcommand{\xxxSubParagraphNoStar}[1]{\oldsubparagraph{#1}\mbox{}}
\fi
\makeatother

\usepackage{longtable,booktabs,array}
\usepackage{calc} % for calculating minipage widths
\usepackage{etoolbox}
\makeatletter
\patchcmd\longtable{\par}{\if@noskipsec\mbox{}\fi\par}{}{}
\makeatother
\IfFileExists{footnotehyper.sty}{\usepackage{footnotehyper}}{\usepackage{footnote}}
\makesavenoteenv{longtable}
\usepackage{graphicx}
\makeatletter
\def\maxwidth{\ifdim\Gin@nat@width>\linewidth\linewidth\else\Gin@nat@width\fi}
\def\maxheight{\ifdim\Gin@nat@height>\textheight\textheight\else\Gin@nat@height\fi}
\makeatother
\setkeys{Gin}{width=\maxwidth,height=\maxheight,keepaspectratio}
\makeatletter
\def\fps@figure{htbp}
\makeatother

\makeatletter
\@ifpackageloaded{caption}{}{\usepackage{caption}}
\AtBeginDocument{%
\ifdefined\contentsname
  \renewcommand*\contentsname{Table of contents}
\else
  \newcommand\contentsname{Table of contents}
\fi
\ifdefined\listfigurename
  \renewcommand*\listfigurename{List of Figures}
\else
  \newcommand\listfigurename{List of Figures}
\fi
\ifdefined\listtablename
  \renewcommand*\listtablename{List of Tables}
\else
  \newcommand\listtablename{List of Tables}
\fi
\ifdefined\figurename
  \renewcommand*\figurename{Figure}
\else
  \newcommand\figurename{Figure}
\fi
\ifdefined\tablename
  \renewcommand*\tablename{Table}
\else
  \newcommand\tablename{Table}
\fi
}
\@ifpackageloaded{float}{}{\usepackage{float}}
\floatstyle{ruled}
\@ifundefined{c@chapter}{\newfloat{codelisting}{h}{lop}}{\newfloat{codelisting}{h}{lop}[chapter]}
\floatname{codelisting}{Listing}

\makeatother
\makeatletter
\@ifpackageloaded{caption}{}{\usepackage{caption}}
\@ifpackageloaded{subcaption}{}{\usepackage{subcaption}}
\makeatother

\ifLuaTeX
  \usepackage{selnolig}  % disable illegal ligatures
\fi
\usepackage[]{natbib}
\usepackage{bookmark}

\IfFileExists{xurl.sty}{\usepackage{xurl}}{} % add URL line breaks if available
\newcommand{\anon}{1}

\usepackage{amsmath}
\usepackage{amssymb}
\usepackage{amsthm}
\usepackage{bm}
\usepackage{multirow}
\newtheorem{lemma}{Lemma}
\newtheorem{theorem}{Theorem}
\newcommand{\Var}{\mathrm{Var}}
\newcommand{\Cov}{\mathrm{Cov}}
\def\E{\mathbb{E}}
\def\bs{\bm{s}}
\def\bt{\bm{t}}
\def\bu{\bm{u}}
\def\bv{\bm{v}}
\def\bw{\bm{w}}
\def\bU{\bm{U}}
\def\bre{\bm{\mathrm{e}}}
\def\brw{\bm{\mathrm{w}}}
\def\bA{\bm{\mathrm{A}}}
\def\bB{\bm{\mathrm{B}}}
\def\bC{\bm{\mathrm{C}}}
\def\rd{\mathrm{d}}
\def\rw{\mathrm{w}}
\def\balpha{\boldsymbol{\alpha}}
\def\bbeta{\boldsymbol{\beta}}
\def\btheta{\boldsymbol{\theta}}
\def\bzeta{\boldsymbol{\zeta}}
\def\bSigma{\boldsymbol{\Sigma}}
\def\tbbeta{\hat{\boldsymbol{\beta}}}
\def\tbtheta{\hat{\boldsymbol{\theta}}}
\def\bQ{\bm{Q}}
\def\tC{\hat{C}}
\def\tV{\hat{V}}
\def\tomega{\hat{\omega}}

\begin{document}

\def\spacingset#1{\renewcommand{\baselinestretch}%
{#1}\small\normalsize} \spacingset{1}

%%%%%%%%%%%%%%%%%%%%%%%%%%%%%%%%%%%%%%%%%%%%%%%%%%%%%%%%%%%%%%%%%%%%%%%%%%%%%%

\if1\anon
{
  \title{\bf On Ignorability of Preferential Sampling in Geostatistics}
    \author{Changqing Lu\\
    Centrum Wiskunde \& Informatica, Amsterdam\\
    and\\
    Ganggang Xu\\
    University of Miami, Coral Gables, FL\\
    and\\
    Junho Yang\\
    Academia Sinica, Taipei\\
    and\\
    Yongtao Guan\\
    The Chinese University of Hong Kong, Shenzhen\\
    }
  % \author{Author 1\thanks{
  %   The authors gratefully acknowledge \textit{please remember to list all relevant funding sources in the version that gives all author information}}\hspace{.2cm}\\
  %   Department of YYY, University of XXX\\
  %   and \\
  %   Author 2 \\
  %   Department of ZZZ, University of WWW}
  \date{}
  \maketitle
} \fi

\if0\anon
{
  \bigskip
  \bigskip
  \bigskip
  \begin{center}
    {\LARGE\bf Title}
\end{center}
  \medskip
} \fi

\bigskip
\begin{abstract}
Preferential sampling has attracted considerable attention in geostatistics since the pioneering work of \citet{diggle2010geostat}. A variety of likelihood-based approaches have been developed to correct estimation bias by explicitly modelling the sampling mechanism. While effective in many applications, these methods are often computationally expensive and can be susceptible to model misspecification. In this paper, we present a surprising finding: some existing non-likelihood-based methods that ignore preferential sampling can still produce unbiased and consistent estimators under the widely used framework of \citet{diggle2010geostat} and its extensions. We investigate the conditions under which preferential sampling can be ignored and develop estimators for both regression and covariance parameters without specifying the sampling mechanism parametrically. Simulation studies demonstrate clear advantages of our approach, including reduced estimation error, improved confidence interval coverage, and substantially lower computational cost. To show the practical utility, we further apply it to a tropical forest data set.
\end{abstract}

\noindent%
\vfill

\newpage
\spacingset{1.8} % DON'T change the spacing!

\section{Introduction}
\label{sec:1}

Geostatistics models spatially continuous phenomena using data observed at discrete locations $\bs_1,\dots,\bs_n$ in a region of interest $S \subset\mathbb{R}^2$. A commonly used formulation is $Z_i(\bs_i) = \bw(\bs_i)^\top \bbeta + Y(\bs_i) + e_i$, where $\bw(\bs_i)\in \mathbb{R}^p$ denotes spatial covariates with associated regression coefficients $\bbeta$, $Y(\bs)$ is a latent zero-mean Gaussian process, and $e_i$ are independent Gaussian errors (nugget effects) with variance $\sigma_e^2$. The primary objectives are to consistently estimate the regression coefficients $\bbeta$, the covariance function of $Y(\bs)$, and the nugget variance $\sigma_e^2$.

In classical geostatistical models, sampling locations are typically assumed to be deterministic or independent from the underlying spatial process, in which case the standard maximum likelihood estimation (MLE) is generally preferred \citep{diggle2019book}. However, \citet{diggle2010geostat} pointed out that, in many applications, the process $Z(\bs)$ may depend on the locations at which it is observed. This phenomenon, termed as preferential sampling, can introduce substantial bias into the standard MLE, necessitating careful methodological adjustments. Recognizing its importance, a large body of research has focused on addressing this issue. For example, \citet{diggle2010geostat} proposed a marked point process framework, where the observed locations are modelled as a realization of a log-Gaussian Cox process \citep[LGCP]{moller1998lgcp}, and the corresponding spatial measurements $Z(\bs)$ are treated as marks generated from a Gaussian process. Within this framework, the dependence between locations and marks is conveniently captured through a parametric relationship between their respective Gaussian random fields, facilitating likelihood-based estimation and inference for all model parameters \citep{dinsdale2019tmb}.

To account for potential preferential sampling, the LGCP-based framework has been widely applied across various disciplines, such as ecology \citep{Pennino2018species} and physical oceanography~\citep{Dinsdale2019application}. Methodologically, \citet{Pati2011bayesian} developed a Bayesian approach for estimating the mean of the mark process, while \citet{Zidek2014spacetime} extended the framework to a space–time context. \citet{Ferreira2015sampling} employed it to guide the selection of new sampling locations, and \citet{Amaral2024spatiallyvarying} examined spatially varying sampling degrees. Within the point process literature, further extensions have been proposed to accommodate specialized data structures, e.g.\ modelling the shared latent field via functional analysis when replicated marked point processes are available \citep{fok2012functional,gervini2020joint,xu2020semi,yin2021row,xu2024bias}. Recently, \citet{Schliep2023weightedcl} and \citet{Hsiao2025inverseintensity} adopted composite likelihood approaches with intensity-related weights to improve parameter estimation for the mark process. However, no theoretical guarantees have yet been established for these methods under preferential sampling.

A major limitation of the likelihood-based approaches is the need to specify the generating mechanism of the preferential sampling, which imposes a parametric formulation on the dependence between the point process and the marks. Consequently, these methods can be susceptible to model misspecification. In this work, we present a surprising finding: under the framework of \citet{diggle2010geostat} and its extensions, some existing non-likelihood-based methods that ignore preferential sampling can still produce unbiased and consistent estimators. We carefully investigate the conditions under which preferential sampling can be ignored. Building on that, we develop estimators for both regression and covariance parameters without specifying a parametric sampling mechanism and establish statistical inference for the former. Our method is therefore applicable to a broader range of problems and, as shown in simulation studies, offers substantial computational gains over existing likelihood-based approaches.

The rest of the paper is organized as follows. Section~\ref{sec:2} introduces the geostatistical model under preferential sampling and outlines the technical conditions for theoretical results. Section~\ref{sec:3} examines the asymptotic behaviour of the least squares estimator for the regression coefficients and discusses the approximation of its asymptotic covariance matrix. In Section~\ref{sec:4}, we develop unbiased estimators for the parametric spatial covariance function and establish their consistency. Section~\ref{sec:5} presents numerical experiments that evaluate the proposed method and compare it with likelihood-based approaches. In Section~\ref{sec:6}, we apply our method to a tropical forest data set to demonstrate its practical utility. Finally, the paper concludes with a discussion of the main findings and potential directions for future research.

\section{Preliminaries and Technicalities}
\label{sec:2}

\subsection{The Geostatistical Model under Preferential Sampling}
\label{sec:2_1}

Recall the classical geostatistical model introduced in Section~\ref{sec:1}
\begin{equation}
    Z(\bs)=\bw(\bs)^{\top}\bbeta+Y(\bs)+e(\bs),
    \label{e:1}
\end{equation}
where $Y(\bs)$ is a zero-mean stationary Gaussian random field on $S$ with variance $\sigma_Y^2$ and covariance function $C_Y(\bs,\bt) = \Cov[Y(\bs), Y(\bt)]$. Let $\|\cdot\|$ denote the Euclidean distance. With slight abuse of notation, we assume $C_Y(\bs,\bt)=C_Y(\|\bs-\bt\|)$ for some function $C_Y(r)$, implying that $Y(\bs)$ is isotropic. A popular choice for $C_Y(r)$ is the Mat\'{e}rn covariance function 
\begin{equation}
    C(r;\btheta)=\sigma^2\frac{2^{1-\nu}}{\Gamma(\nu)}\left(\sqrt{2\nu}\frac{r}{\phi}\right)^\nu B_\nu\left(\sqrt{2\nu}\frac{r}{\phi}\right),
    \label{e:2}
\end{equation}
where $\phi$ and $\nu$ are the range and smoothness parameters, $B_\nu$ is the modified Bessel function of the second kind of order $\nu$, and $\Gamma$ is the gamma function.

To generalize the preferential sampling framework of \citet{diggle2010geostat}, we assume that the sampling locations are generated from an LGCP denoted by $N$ and defined also on $S$ with latent intensity
$
    \lambda(\bs)=\lambda_0(\bs)\exp[X(\bs)],
$
where $\lambda_0(\bs)$ is a baseline intensity and $X(\bs)$ is a zero-mean Gaussian random field with covariance function $C_X(\bs,\bt)=\Cov[X(\bs),X(\bt)]$. Consequently, the marginal first- and second-order intensity functions of $N$ are given by
\begin{equation*}
    \rho(\bs)=\mathbb{E}[\lambda(\bs)]=\lambda_0(\bs)\exp(\sigma_X^2/2)
\end{equation*}
and
\begin{equation*}
    \rho_2(\bs,\bt)=\mathbb{E}[\lambda(\bs)\lambda(\bt)]=\rho(\bs)\rho(\bt)\exp[C_X(\bs,\bt)].
\end{equation*}
To introduce the dependence between $X(\bs)$ and $Y(\bs)$, \citet{diggle2010geostat} assumed a constant baseline $\lambda_0(\bs)=\lambda_0$ and a proportional relationship $X(\bs)=\gamma Y(\bs)$ for some $\lambda_0 > 0$ and $\gamma \neq 0$. We extend this setting by allowing an arbitrary isotropic cross-covariance function $\Cov[X(\bs), Y(\bt)] = C_{XY}(||\bs - \bt||)$, without imposing a specific parametric form.

\subsection{Asymptotic Regime and Technical Conditions}
\label{sec:2_2}

Following \citet{diggle2010geostat}, we analyze the model (\ref{e:1}) within the framework of marked point process theory. Specifically, we adopt the standard increasing-domain regime (see, e.g.\ \citealp{guan2007asymptotics,xu2019stochastic,xu2023semiparametric}). Suppose that the observations of $Z(\bs)$ and $N$ are collected over a sequence of region $S_n$ that expand to $\mathbb{R}^2$ as $n\to \infty$. Let $\partial S_n$ denote the boundary of $S_n$ with perimeter $|\partial S_n|$. We assume that, for every $n\geq 1$, 
\begin{equation*}
    c_1n^2\leq|S_n|\leq c_2n^2,\quad c_1n\leq|\partial S_n|\leq c_2n,\quad \mathrm{for}\ \mathrm{some}\ 0<c_1\leq c_2<\infty. \tag{C1}
\end{equation*}
This condition ensures that $S_n$ grows in all directions and that $\partial S_n$ is not too irregular. 

To quantify the spatial dependence, we recall the definition of strong mixing coefficients \citep{rosenblatt1956mixing}. Define
\begin{equation*}
\begin{split}
    \alpha(q; k) \equiv \sup&\left\{|P (S_1 \cap S_2) - P(S_1)P(S_2)|:S_1 \in \mathcal{F}(T_1), S_2 \in \mathcal{F}(T_2),\right.\\
    &\left.T_1, T_2 \subset \mathbb{R}^2, |T_1| = |T_2| \leq q, d(T_1, T_2) \geq k\right\},
\end{split}
\end{equation*}
where $\mathcal{F(\cdot)}$ denotes the $\sigma$-algebra generated by the random events of $N$ that are in a subset of $\mathbb{R}^2$, and $d(T_1, T_2)$ denotes the maximal distance between $T_1$ and $T_2$. We assume that
\begin{equation*}
    \sup_q \alpha(q; k)/q = O(k^{-\epsilon}),\quad \mathrm{for}\ \mathrm{some}\ \epsilon>2, \tag{C2}
\end{equation*}
which requires the dependence between any two fixed sets decaying to zero at a polynomial rate of the inter-set distance $k$, while the decay rate also depends on the size of the sets $q$. An LGCP, as assumed in Section~\ref{sec:2_1}, satisfies this condition. 

In addition to (C1)--(C2), we impose regularity conditions on the spatial covariates $\bw(\bs)$ and the baseline intensity $\lambda_0(\bs)$:
\begin{equation*}
    \sup_{\bs\in \mathbb{R}^2} \|\bw(\bs)\|< \infty,\quad \sup_{\bs\in \mathbb{R}^2} \left|\lambda_0(\bs)\right|< \infty, \tag{C3}
\end{equation*}
on the covariance functions of the Gaussian random fields $X(\bs)$, $Y(\bs)$ and $e(\bs)$: 
\begin{equation*}
    \sup_{\bs,\bt\in \mathbb{R}^2} \left|C_X(\|\bs-\bt\|)\right|< \infty, \quad \sup_{\bs,\bt\in \mathbb{R}^2} \left|C_Y(\|\bs-\bt\|)\right|< \infty, \quad \sigma_e^2<\infty, \tag{C4}
\end{equation*}
and on the cross-covariance function between $X(\bs)$ and $Y(\bs)$: 
\begin{equation*}
    \sup_{\bs,\bt\in \mathbb{R}^2} \left|C_{XY}(\|\bs-\bt\|)\right|< \infty. \tag{C5}
\end{equation*}
Moreover, we assume that $C_Y(\|\bs-\bt\|)$ and $C_{XY}(\|\bs-\bt\|)$ are absolutely integrable everywhere on $\mathbb{R}^2$:
\begin{equation*}
    \sup_{\bs\in \mathbb{R}^2} \int_{\mathbb{R}^2} |C_Y(\|\bs-\bt\|)|\rd\bt< \infty,\quad
    \sup_{\bs\in \mathbb{R}^2} \int_{\mathbb{R}^2} |C_{XY}(\|\bs-\bt\|)|\rd\bt< \infty.
 \tag{C6}
\end{equation*}

Finally, we require that the $p\times p$ matrix
\begin{equation*}
    \int_{S_n} \rho(\bs)\bw(\bs)\bw(\bs)^\top\rd\bs \tag{C7}
\end{equation*}
is invertible, a technical condition necessary for deriving the asymptotic covariance matrix in Theorem~\ref{theorem:1}. Furthermore, for Theorem~\ref{theorem:3}, we assume that the parametric semi-variogram function defined in Section~\ref{sec:4},
\begin{equation*}
    \zeta(\|\bs-\bt\|;\theta)>0, \tag{C8}
\end{equation*}
has continuous partial derivatives with respect to $\btheta$. We also assume that the weight function $\rw(\bs,\bt)$ used in the objective functions $Q_{MC}(\btheta)$ and $Q_{CL}(\btheta)$ satisfies
\begin{equation*}
    \sup_{\bs,\bt\in \mathbb{R}^2}|\rw(\bs,\bt)|<\infty, \tag{C9}
\end{equation*}
and that the corresponding minimizers
\begin{equation*}
    \tbtheta_{MC}=\arg\min_{\btheta} Q_{MC}(\btheta),\quad \tbtheta_{CL}=\arg\min_{\btheta} Q_{CL}(\btheta)
\tag{C10}
\end{equation*}
on $S_n$ are unique.

It is worthy noting that, unlike most studies on marked point processes, we do not need to introduce an additional reference measure on the mark space. This is because, within our preferential sampling framework, $Z(\bs)$ is modelled as a spatial Gaussian process, whose moments, conditional on N, can be characterized through the joint correlations, i.e.\ via the cross-covariance function.

\section{Statistical Inference for the Regression Coefficients}
\label{sec:3} 

A primary objective with the geostatistical model (\ref{e:1}) is to estimate the regression coefficients $\bbeta$ for capturing spatial heterogeneity. Consider the least squares estimator
\begin{equation}
    \tbbeta=\left[\sum_{\bs \in N}\bw(\bs)\bw(\bs)^\top\right]^{-1}\sum_{\bm{s}\in N}\bw(\bs)Z(\bs).
    \label{e:3}
\end{equation}
In the absence of preferential sampling, $\tbbeta$ is a consistent estimator with well-established asymptotic properties \citep{banerjee2014book}. However, when the sampling locations are a realization of a point process that depends on $Z(\bs)$, the cross-correlation in between must be accounted for.

By applying the Campbell's theorem and the Stein's lemma, the expectations of the denominator and the numerator in (\ref{e:3}) read
\begin{equation*}
    \E\left[\sum_{\bs\in N}\bw(\bs)\bw(\bs)^\top\right]=\int_S \rho(\bs)\bw(\bs)\bw(\bs)^\top\rd\bs
\end{equation*}
and
\begin{equation*}
    \E\left[\sum_{\bs\in N} \bw(\bs) Z(\bs)\right]=\left[\int_S \rho(\bs)\bw(\bs)\bw(\bs)^\top\rd\bs\right]\bbeta+C_{XY}(0)\int_S \rho(\bs)\bw(\bs)\rd\bs.
\end{equation*}
Assuming that the first element of $\bw(\bs)$ corresponds to the intercept term, we can use the results above to analyze the asymptotic behaviour of the least squares estimator (\ref{e:3}) under the increasing-domain regime, as stated in the following theorem.

\begin{theorem}
\label{theorem:1}
    Let $\bbeta_0$ and $\tbbeta_n$ denote the true regression coefficients and the parameters estimated by (\ref{e:3}) on $S_n$. Under conditions (C1)--(C7), as $n\to \infty$, 
    \begin{equation*}
        |S_n|^{1/2}\bSigma_n^{-1/2}(\tbbeta_n-\bbeta^*_0)\xrightarrow{d}\mathrm{N}(\bm{0},\bm{I}_p),
    \end{equation*}
    where $\bm{I}_p$ is the $p\times p$ identity matrix, $\bbeta^*_0=\bbeta_0+C_{XY}(0)\balpha$ with  $\balpha=(1,0,\dots,0)^\top$, and
    \begin{equation*}
    \begin{split}
        \bSigma_n&=|S_n|\bA_n^{-1}(\bB_n+\bC_n)\bA_n^{-1},\quad \bA_n=\int_{S_n} \rho(\bs)\bw(\bs)\bw(\bs)^\top\rd\bs,\\ 
        \bB_n&=\int_{S_n} \rho(\bs)\left(\sigma_Y^2+\sigma_e^2\right)\bw(\bs)\bw(\bs)^\top\rd\bs,\\
        \bC_n&=\int_{S_n}\int_{S_n} \rho_2(\bs,\bt)\left[C_Y(\|\bs-\bt\|)+C_{XY}(\|\bs-\bt\|)^2\right]
        \bw(\bs)\bw(\bt)^\top\rd\bs\rd\bt.
    \end{split}
    \end{equation*}
\end{theorem}
\begin{proof}
    The proof is provided in the Supplementary Material.
\end{proof}

Theorem~\ref{theorem:1} shows that, despite preferential sampling, the least squares estimator~\eqref{e:3} remains unbiased for all regression coefficients in~$\bbeta$, except for the intercept term. If a consistent estimator of~$C_{XY}(0)$ is available, the intercept bias can be readily corrected via $\tilde{\bbeta} = {\tbbeta} - \hat{C}_{XY}(0)\balpha$. Moreover, this theorem also enables valid statistical inference for all parameters in~$\bbeta$, provided that consistent estimators of the sill $\omega=\sigma_Y^2+\sigma_e^2$, the spatial covariance function $C_Y(r)$ and the cross-covariance function $C_{XY}(r)$ are available.

Typically, we estimate the sill $\omega$ via the moment-based estimator:
\begin{equation}
    \tomega={\frac{1}{|N|}}\sum_{\bs\in N}\left[Z(\bs)-\bw(\bs)^\top\tbbeta\right]^2,
    \label{e:4}
\end{equation} 
where $|N|$ is the number of observations. Recalling the definition of the semi-variogram \citep{banerjee2014book}, $V_Y(r)=\omega-C_Y(r)$, estimating $C_Y(r)$ then reduces to estimating
\newpage
$V_Y(r)$, for which we use the kernel smoothing estimator:
\begin{equation}
    \tV_Y(r)=\frac{\mathop{\sum\sum}^{\neq}_{\bs,\bt\in N}\left\{\left[Z(\bs)-\bw(\bs)^\top\tbbeta\right]-\left[Z(\bt)-\bw(\bt)^\top\tbbeta\right]\right\}^2K_h(\|\bs-\bt\|-r)}
    {2\mathop{\sum\sum}^{\neq}_{\bs,\bt\in N}K_h(\|\bs-\bt\|-r)},
    \label{e:5}
\end{equation}
where $K_h(\cdot)=h^{-1}K(\cdot/h)$ with kernel function $K(\cdot)$ and bandwidth $h$. The superscript $\ne$ indicates that the sum is taken over pairs of distinct points in $N$. Similarly, the cross-covariance $C_{XY}(r)$ can be estimated as 
\begin{equation}
    \tC_{XY}(r)=\frac{\mathop{\sum\sum}^{\neq}_{\bs,\bt\in N}\left[Z(\bs)-\bw(\bs)^\top\tbbeta\right]K_h(\|\bs-\bt\|-r)}{\mathop{\sum\sum}^{\neq}_{\bs,\bt\in N}K_h(\|\bs-\bt\|-r)}.
    \label{e:6}
\end{equation}
Assuming that the semivariogram $V_Y(r)$ and the cross-covariance $C_{XY}(r)$ are smooth in a neighborhood of $r$, the following theorem establishes consistency of the three estimators (\ref{e:4})--(\ref{e:6}) under preferential sampling.

\begin{theorem}
\label{theorem:2}
 Under conditions (C1)--(C6), as $h_n\to 0$ and $|S_n|h_n\to \infty$, it holds that, for every $r>0$, $|\tomega-\omega|=O_p(|S_n|^{-1/2})$, $ |\tV_Y(r)-V_Y(r)|= O_p[h_n+(|S_n|h_n)^{-1/2}]$, and $ |\tC_{XY}-C_{XY}(r)|= O_p[h_n+(|S_n|h_n)^{-1/2}].$
\end{theorem}
\begin{proof}
    The proof is provided in the Supplementary Material.
\end{proof}

Theorem~\ref{theorem:2} reveals a surprising result: even in the presence of preferential sampling, the sill $\omega$, the semi-variogram function $V_Y(r)$ and the cross-covariance function $C_{XY}(r)$ can still be consistently estimated using classical moment-based estimator (\ref{e:4}) and kernel-based estimators (\ref{e:5}) and (\ref{e:6}), without having to make parametric assumptions on $X(\bs)$ nor the preferential sampling mechanism. Therefore, we can simply plug $\tomega$, $\tC_Y(r)=\tomega-\tV_Y(r)$ and $\tC_{XY}$ back into the asymptotic covariance matrix derived in Theorem~\ref{theorem:1} to enable the inference for $\tbbeta$. In the simulation study in Section~\ref{sec:5_2}, we demonstrate that this approach yields valid confidence intervals.

\section{Unbiased Estimation of a Parametric Spatial Covariance Function}
\label{sec:4} 

Although the nonparametric estimator $\tV_Y(r)$ consistently estimates the semi-variogram $V_Y(r)$ under the preferential sampling framework described in Section~\ref{sec:2}, it is often desirable in geostatistics to fit a parametric model for $V_Y(r)$, and equivalently for the spatial covariance $C_Y(r)$. For instance, one may assume $V_Y(r)=\zeta(r;\btheta)$ with $\zeta(r;\btheta)=\sigma_e^2+\sigma_Y^2[1-\exp(-r/\phi_Y)]$, which corresponds to a special case of the Mat\'{e}rn class (\ref{e:2}) with $\nu=0.5$ and $\btheta=(\sigma_Y^2,\phi_Y,\sigma_e^2)$.

To estimate $\btheta$, we propose to minimize the weighted minimum contrast objective function:
\begin{equation}
\int_{S}\int_{S}\rw(\bs,\bt)\rho_2(\bs,\bt)\left[V_Y(\|\bs-\bt||)-\zeta(\|\bs-\bt\|;\btheta)\right]^2\rd\bs\rd\bt,
    \label{e:7}
\end{equation}
where $\rw(\bs,\bt)$ is a nonnegative weight function. It is straightforward to show that, under preferential sampling, minimizing (\ref{e:7}) is equivalent to minimizing
\begin{equation*}
\E\left[\mathop{\sum\sum}^{\neq}_{\bs,\bt\in N}\rw(\bs,\bt)\left\{\frac{1}{2}\left[Z^*(\bs)-Z^*(\bt)\right]^2-\zeta(\|\bs-\bt\|;\btheta)\right\}^2\right],
\end{equation*}
where ${Z}^*(\bs)=Z(\bs)-\bw(\bs)^\top\bbeta_0^*$ with $\bbeta_0^*$ defined in Theorem~\ref{theorem:1}. Note that, by Theorem~\ref{theorem:1}, $\tbbeta$ is a consistent estimator for $\bbeta_0^*$. We therefore propose to minimize the following function 
\begin{equation}
    Q_{MC}(\btheta)=\mathop{\sum\sum}^{\neq}_{\bs,\bt\in N}\rw(\bs,\bt)\left\{\left[\hat{Z}(\bs)-\hat{Z}(\bt)\right]^2-2\zeta(\|\bs-\bt\|;\btheta)\right\}^2,
    \label{e:8}
\end{equation}
where $\hat{Z}(\bs)=Z(\bs)-\bw(\bs)^\top\tbbeta$. Assuming that $\zeta(\|\bs-\bt\|;\btheta)$ is differentiable with respect to $\btheta$ and denoting the partial derivative by $\bzeta^{(1)}(\|\bs-\bt\|;\btheta)$, minimizing (\ref{e:8}) reduces to solving the estimating equation
\begin{equation}
    \bQ_{MC}^{(1)}(\btheta)=\mathop{\sum\sum}^{\neq}_{\bs,\bt\in N}\rw(\bs,\bt)\bzeta^{(1)}(\|\bs-\bt\|;\btheta)\left\{\left[\hat{Z}(\bs)-\hat{Z}(\bt)\right]^2-2\zeta(\|\bs-\bt\|;\btheta)\right\}=\bm 0.
    \label{e:9}
\end{equation}

Alternatively, note that, for any pair of distinct locations $(\bs,\bt)$, the difference $[Z^*(\bs)-Z^*(\bt)]\sim \mathrm{N}\left[0,2\zeta(\|\bs-\bt\|;\btheta)\right]$. Hence, we can also minimize a weighted negative composite likelihood objective function:
\begin{equation}
    Q_{CL}(\btheta)=\mathop{\sum\sum}^{\neq}_{\bs,\bt\in N}
    \rw(\bs,\bt)\left\{\frac{\left[\hat{Z}(\bs)-\hat{Z}(\bt)\right]^2}{2\zeta(\|\bs-\bt\|;\btheta)}+\log\left[\zeta(\|\bs-\bt\|;\btheta)\right]\right\},
    \label{e:10}
\end{equation}
which is equivalent to to solving the estimating equation 
\begin{equation}
    \bQ_{CL}^{(1)}(\btheta)=\mathop{\sum\sum}^{\neq}_{\bs,\bt\in N}\rw(\bs,\bt)\frac{\bzeta^{(1)}(\|\bs-\bt\|;\btheta)}{\zeta(\|\bs-\bt\|;\btheta)^2}\left\{\left[\hat{Z}(\bs)-\hat{Z}(\bt)\right]^2-2\zeta(\|\bs-\bt\|;\btheta)\right\}=\bm 0.
    \label{e:11}
\end{equation}
A good choice of $\rw(\bs,\bt)$ can improve the efficiency of the resulting estimators. In our simulation study, we adopt the following weight:
\begin{equation*}
    \rw(\bs,\bt)=\frac{1(\|\bs-\bt\|\leq R)}{2\pi \|\bs-\bt\| |S\cap(S-\bs+\bt)|},
\end{equation*}
where $1(\cdot)$ is the indicator function, $|S\cap(S-\bs+\bt)|$ is the overlap area between $S$ and its translation by $\bs-\bt$, and $R$ is a predefined constant of spatial dependence range. 

\begin{theorem}
\label{theorem:3}
    Let $\btheta_0$, $\tbtheta_{n,MC}$ and $\tbtheta_{n,CL}$ denote the true parameters in $\zeta(\|\bs-\bt\|;\btheta)$ and the estimators by minimizing (\ref{e:8}) and (\ref{e:10}) on $S_n$. Under conditions (C1)--(C6) and (C8)--(C10), as $n\to\infty$, it holds that $\|\tbtheta_{n,MC}-\btheta_0\|=O_p(|S_n|^{-1/2})$ and $\|\tbtheta_{n,CL}-\btheta_0\|=O_p(|S_n|^{-1/2})$.
\end{theorem}
\begin{proof}
    The proof is provided in the Supplementary Material.
\end{proof}

Theorem~\ref{theorem:3} establishes consistency of the proposed estimators for the parameters $\btheta$ under preferential sampling. This result holds because, when the sampling locations follow an LGCP and the spatial process $Y(\bs)$ is Gaussian, the two estimating equations (\ref{e:9}) and (\ref{e:11}) remain unbiased, regardless of the strength of cross-correlation, by arguments in analogy to those for Theorem~\ref{theorem:2}.

\section{Simulation Study}
\label{sec:5}

To evaluate the proposed approaches, we conduct simulations under two preferential sampling scenarios. Experiments are performed on $S=[0,n] \times [0,n]$ with covariates $\bw(\bs)=[1, w(\bs)]$, where $w(\bs)$ is a fixed realization from a zero-mean Gaussian process with covariance function $\exp(-10r)$. The model is given by
$
    Z(\bs)=\beta_0+\beta_1 w(\bs)+Y(\bs)+e(\bs)
$
with $\beta_0=\beta_1=1$ and $\sigma_e^2=0.1$. For comparison, we include the standard MLE and the adapted likelihood-based method implemented using the template model builder (TMB) \citep{dinsdale2019tmb}. Estimation performance is measured via bias, standard error (SdErr) and root mean squared error (RMSE). 

In scenario 1, we adopt the same model as in \citet{diggle2010geostat}. We model $Y(\bs)$ by a stationary Gaussian process with mean zero and Mat\'{e}rn covariance $C_Y(r;\sigma_Y^2,\phi_Y,\nu_Y)$. The sampling locations are generated from $N$ with latent intensity $\lambda(\bs)=\exp[\gamma_0+X(\bs)]$ and $X(\bs)=\gamma Y(\bs)$, where $\gamma_0$ is chosen to ensure that the expected number of observations per unit square is 400. This implies $C_{XY}(r)=\gamma C_Y(r;\sigma_Y^2,\phi_Y,\nu_Y)$. We fix $\sigma_Y^2=1$ and $\gamma=1$, and vary $\phi_Y$ and $\nu_Y$.
In scenario 2, we consider a more general setup where $X(\bs)$ and $Y(\bs)$ are stationary Gaussian processes with means zero, marginal covariances $C_{X}(r;\sigma_{X}^2,\phi_X,\nu_X)$ and $C_{Y}(r;\sigma_{Y}^2,\phi_Y,\nu_Y)$, and cross-covariance $C_{XY}(r;\sigma_{XY}^2,\phi_{XY},\nu_{XY})$. %Validity of this bivariate Gaussian process is ensured by imposing constraints on the cross-covariance parameters.

\begin{table}[t]
    \caption{Parameter estimates and computation time on simulations in Section~\ref{sec:5_1}.}
    \label{tab:1}
    \centering
    \scriptsize
    \begin{tabular}{|l|ccccccc|}
    \hline
         & \textbf{Method} & $\beta_0$ & $\beta_1$ & $\sigma_Y^2$ & $\phi_Y$ & $\sigma_e^2$ & Time (sec)\\
    \hline
        & True & 1 & 1 & 1 & 0.05 & 0.1 & -\\
        \cline{2-8}
        \multirow{4}{*}{\shortstack{\textbf{Scenario 1}\\ ($\phi$=0.05)}} & MLE & 1.39(0.16) & 0.98(0.06) & 0.83(0.13) & 0.0440(0.0069) & 0.101(0.023) & 4.7(0.6)\\
        & TMB & 0.96(0.16) & 0.99(0.04) & 0.98(0.17) & 0.0524(0.0073) & 0.101(0.023) & 727.1(125.3)\\
        & CL & 1.06(0.26) & 0.98(0.15) & 1.01(0.26) & 0.0495(0.0121) & 0.102(0.035) & 1.3(0.2)\\
        & MC & 1.06(0.26) & 0.98(0.15) & 1.03(0.30) & 0.0523(0.0151) & 0.103(0.047) & 1.5(0.4)\\
    \hline
        & True & 1 & 1 & 1 & 0.1 & 0.1 & -\\
        \cline{2-8}
        \multirow{4}{*}{\shortstack{\textbf{Scenario 1}\\ ($\phi$=0.1)}} & MLE & 1.19(0.28) & 1.00(0.04) & 0.84(0.20) & 0.0911(0.0172) & 0.099(0.013) & 4.5(0.8)\\
        & TMB & 0.91(0.29) & 1.00(0.04) & 0.92(0.23) & 0.0981(0.0163) & 0.098(0.013) & 775.6(89.6)\\
        & CL & 1.09(0.41) & 1.00(0.16) & 0.97(0.37) & 0.0989(0.0334) & 0.102(0.021) & 2.0(0.3)\\
        & MC & 1.09(0.41) & 1.00(0.16) & 0.97(0.40) & 0.0990(0.0416) & 0.096(0.036) & 1.9(0.4)\\
    \hline
        & True & 1 & 1 & 1 & 0.1 & 0.1 & -\\
        \cline{2-8}
        \multirow{4}{*}{\textbf{Scenario 2}}
        & MLE & 1.39(0.25) & 1.00(0.04) & 0.77(0.20) & 0.0923(0.0177) & 0.101(0.012) & 4.5(0.7)\\
        & TMB & 1.07(0.25) & 1.00(0.06) & 0.80(0.18) & 0.0784(0.0137) & 0.089(0.011) & 1142.3(143.9)\\
        & CL & 1.21(0.29) & 1.01(0.15) & 0.92(0.38) & 0.1006(0.0298) & 0.101(0.018) & 2.1(0.5)\\
        & MC & 1.21(0.29) & 1.01(0.15) & 0.92(0.41) & 0.1030(0.0382) & 0.098(0.032) & 1.9(0.4)\\
    \hline
    \end{tabular}
\end{table}

\subsection{Comparison of Estimation Results}
\label{sec:5_1}

To compare the performance of different approaches, we run 100 simulations per scenario on $S=[0,1]\times [0,1]$. In scenario 1, we fix $\nu_Y=1$ and vary $\phi_Y$ between $0.05$ and $0.1$. In scenario 2, we set $\sigma_{X}^2=1.8, \sigma_{XY}^2=1$, $\sigma_{Y}^2=1$ with $\nu_X=0.5, \nu_{XY}=0.75, \nu_Y=1$ and $\phi_X=0.05, \phi_{XY}=0.07, \phi_Y=0.1$. We choose $R=4\phi_Y$ and report the estimates for the regression and covariance parameters across the 100 runs, along with the computation time. The results are summarized in Table~\ref{tab:1}. Our proposed method consistently yields approximately unbiased parameter estimates in all settings. In contrast, MLE exhibits large bias due to its disregard of preferential sampling. TMB performs reasonably well in scenario 1, where the model is correctly specified, but its performance deteriorates notably in scenario 2, where the model is misspecified. Importantly, our method remains effective under both forms of cross-correlation, demonstrating superior robustness. It is also computationally efficient, specifically, several times faster than MLE and up to thousands of times faster than TMB. Among the parameters, $\beta_1$ and $\sigma_e^2$ are estimated with the highest accuracy. Our method exhibits slightly higher variance compared to MLE and TMB, which is not surprising because the latter two approaches are both likelihood-based methods. Between the minimum contrast (MC) and composite likelihood (CL) estimators, the latter has smaller variances than the former. 

\begin{figure}[t]
\centering
\includegraphics[width=0.32\textwidth]{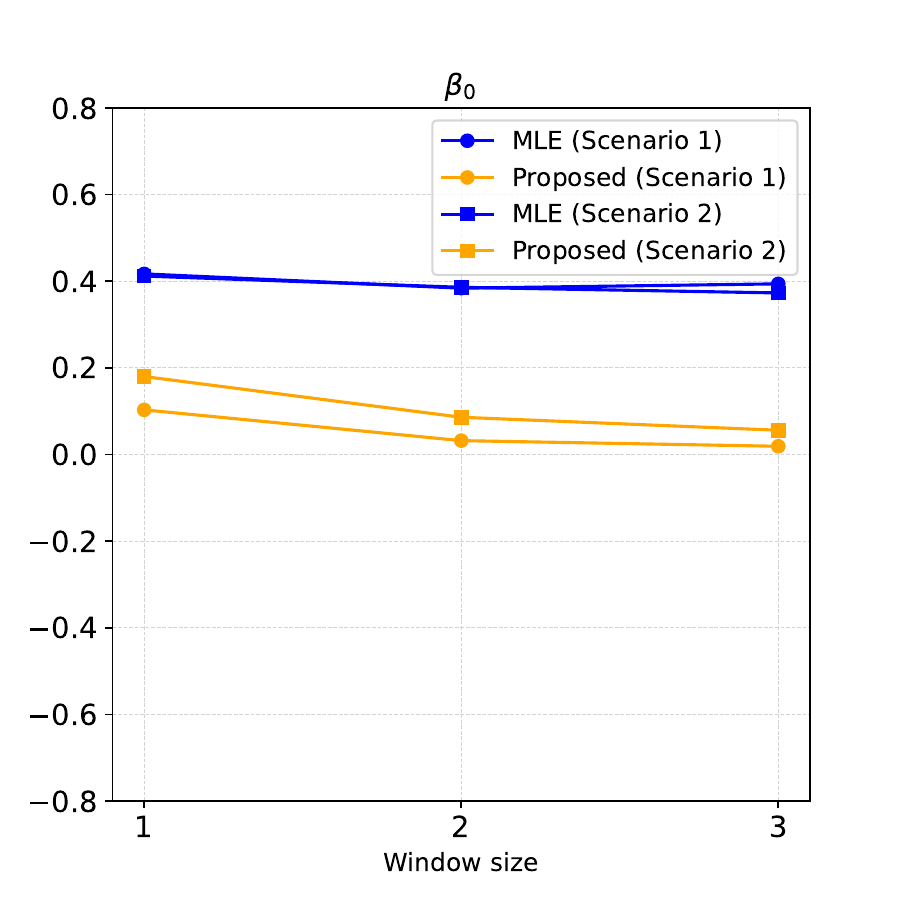}
\includegraphics[width=0.32\textwidth]{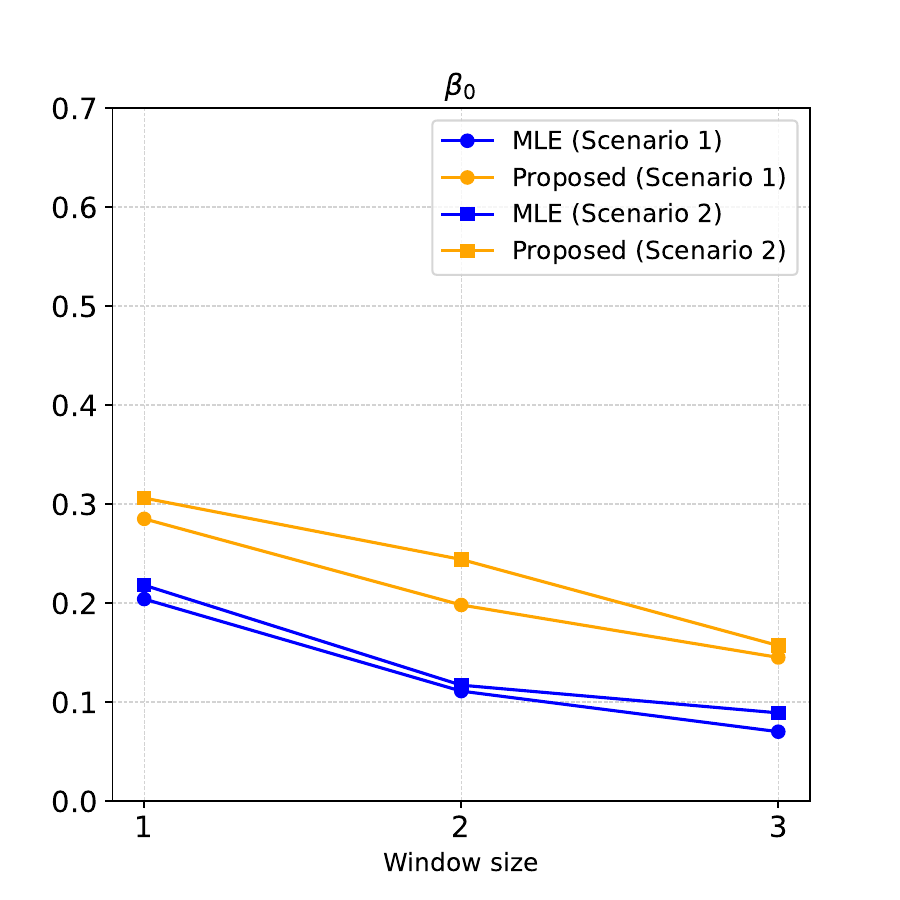}
\includegraphics[width=0.32\textwidth]{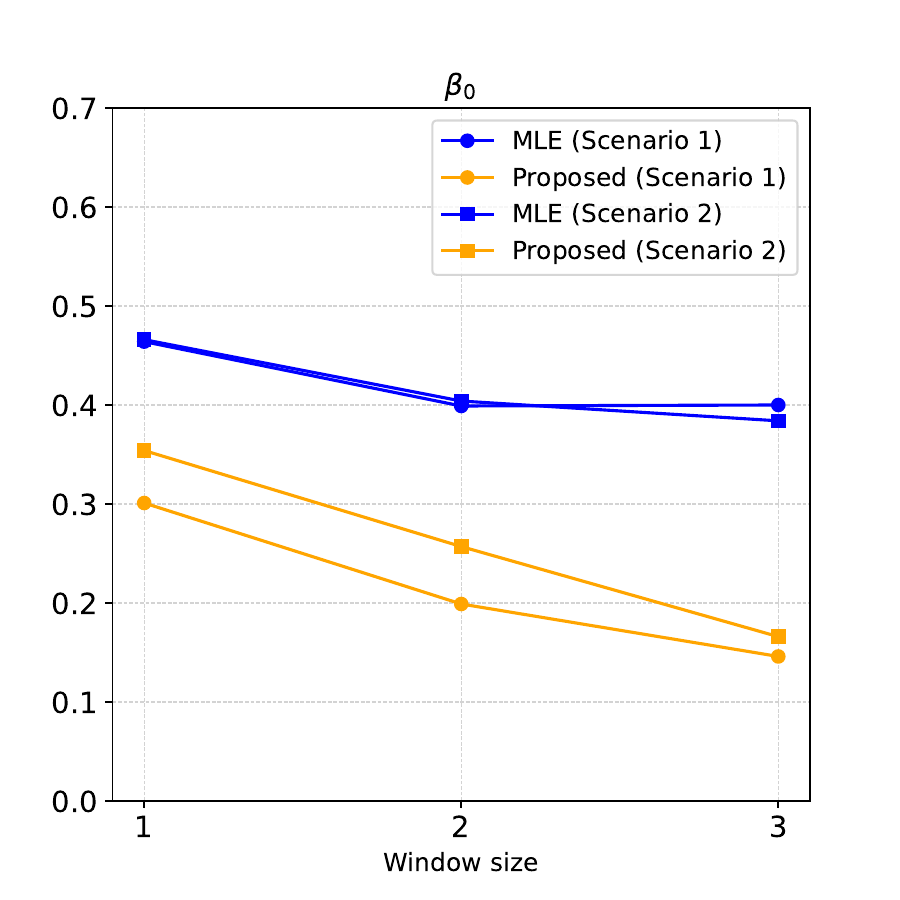}
\includegraphics[width=0.32\textwidth]{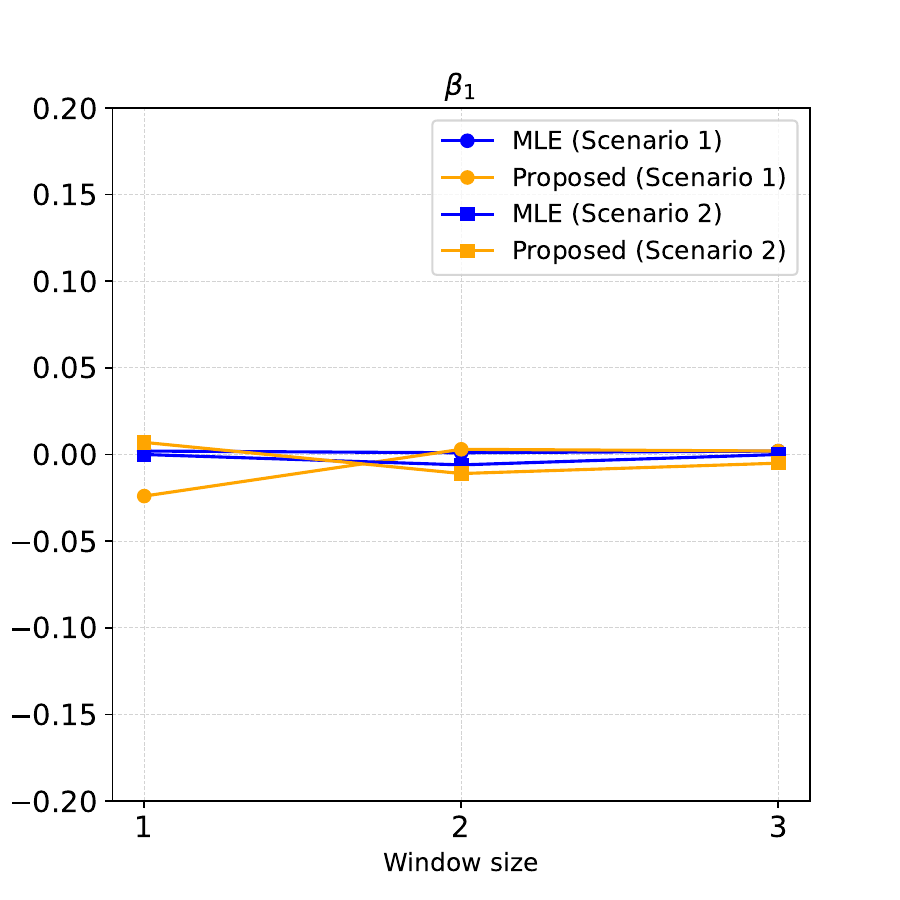}
\includegraphics[width=0.32\textwidth]{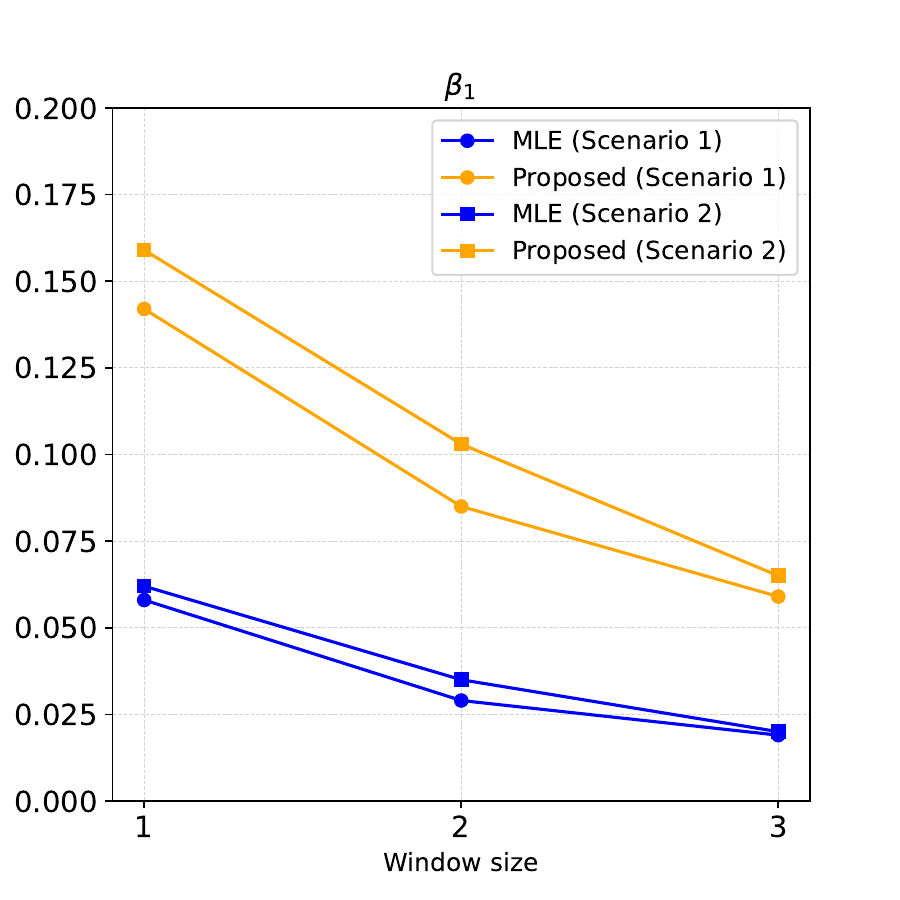}
\includegraphics[width=0.32\textwidth]{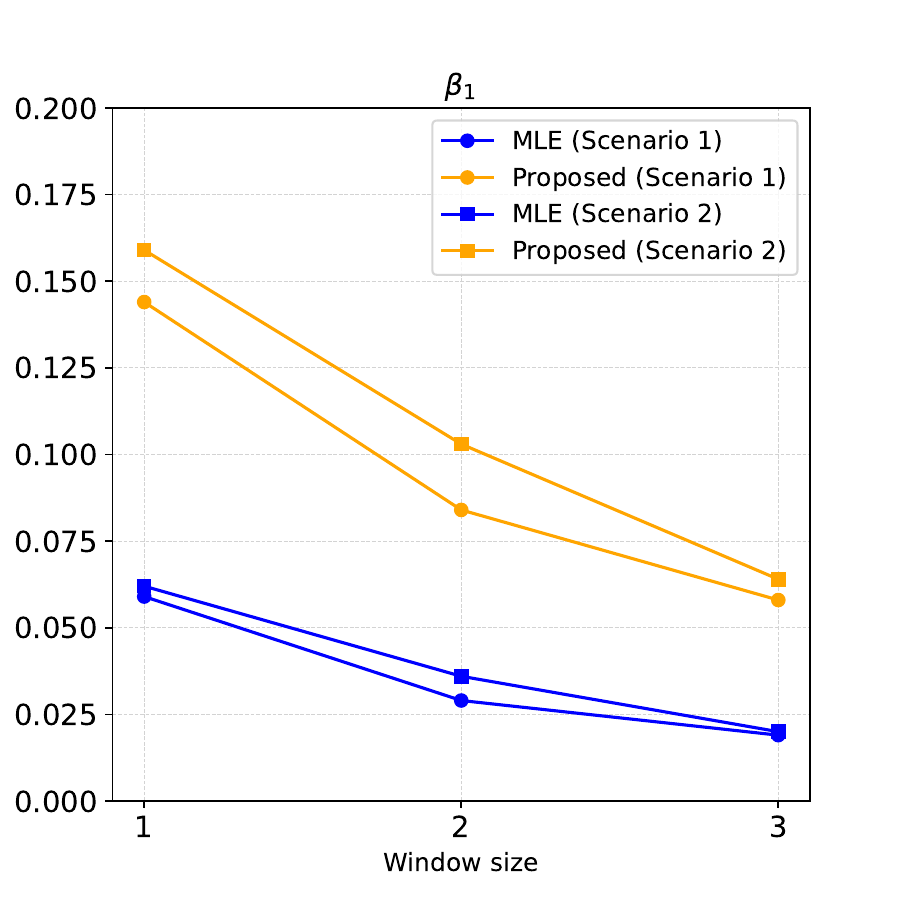}
\caption{Bias (left), SdErr (middle) and RMSE (right) of the estimated regression coefficients on simulations in Section~\ref{sec:5_2}.}\label{fig:1}
\end{figure}

\begin{figure}[t]
\centering
\includegraphics[width=0.32\textwidth]{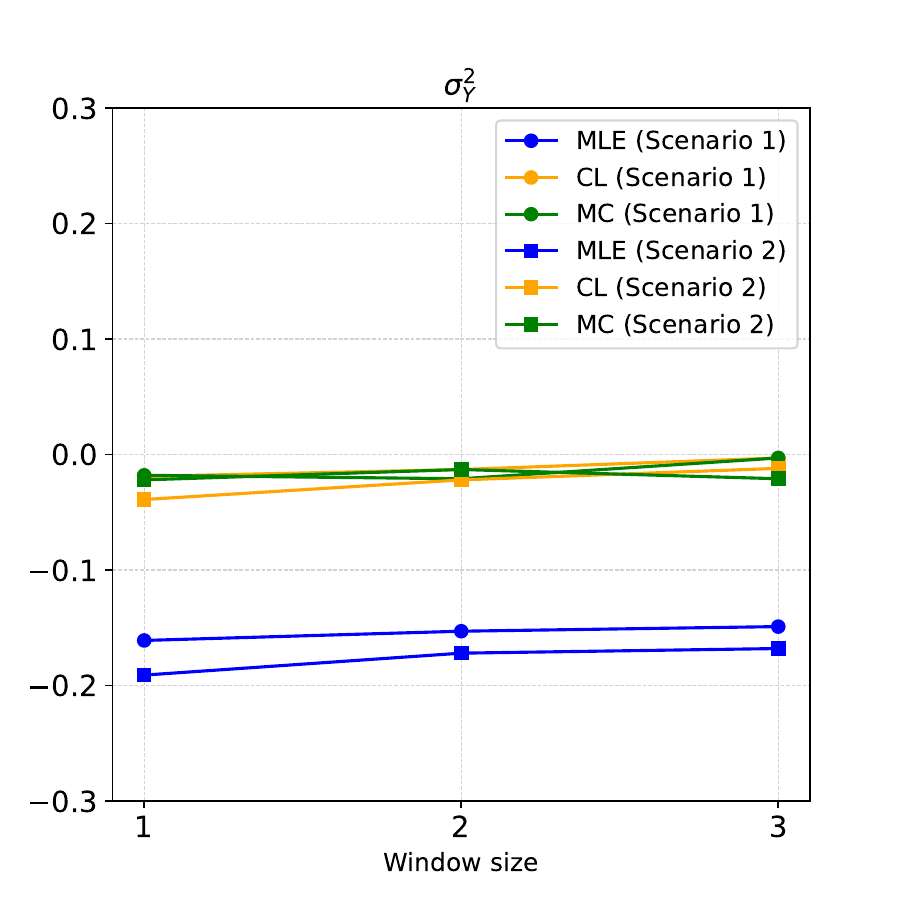}
\includegraphics[width=0.32\textwidth]{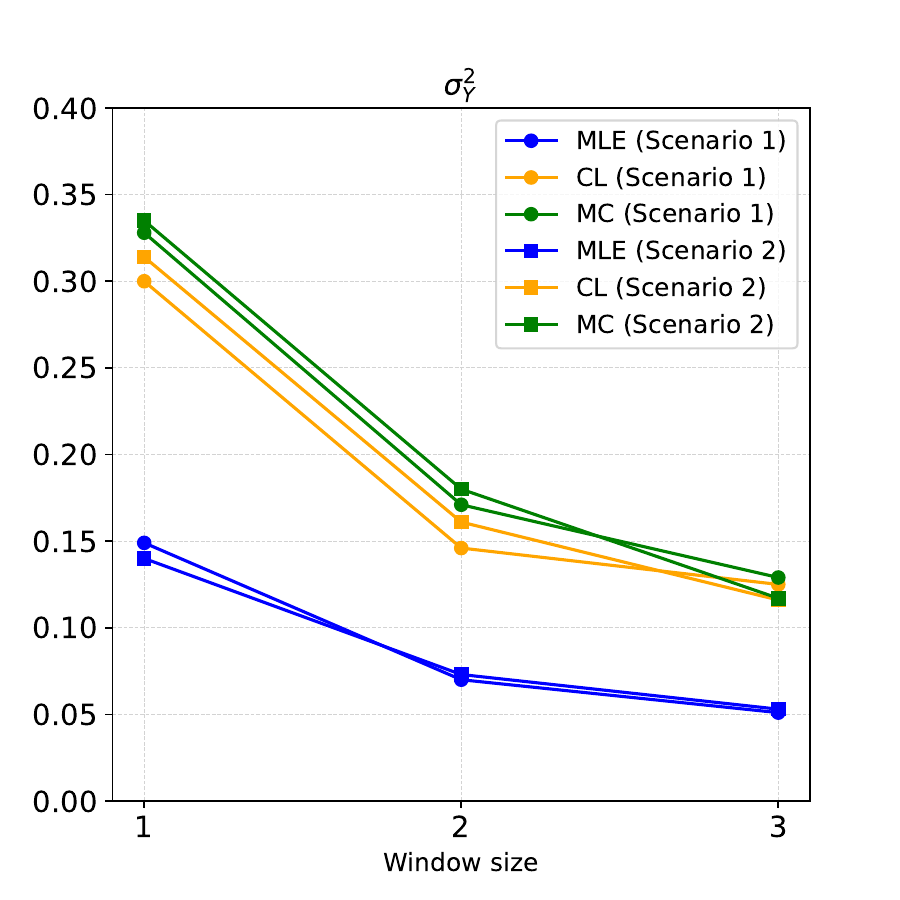}
\includegraphics[width=0.32\textwidth]{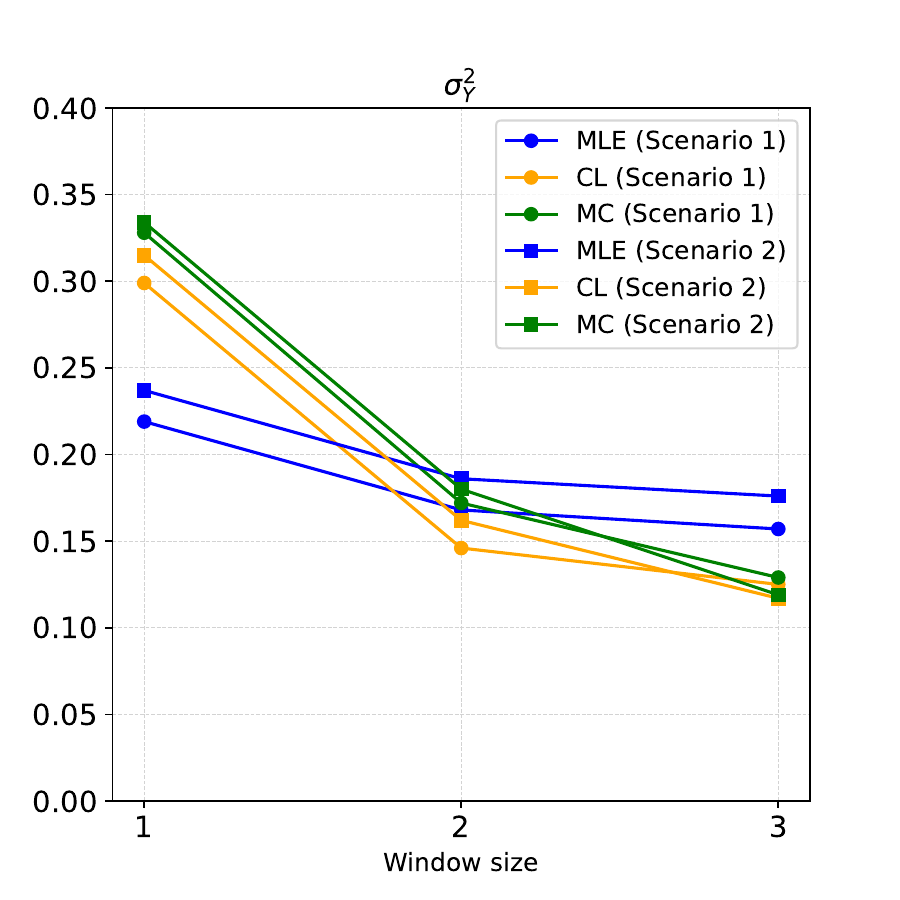}
\includegraphics[width=0.32\textwidth]{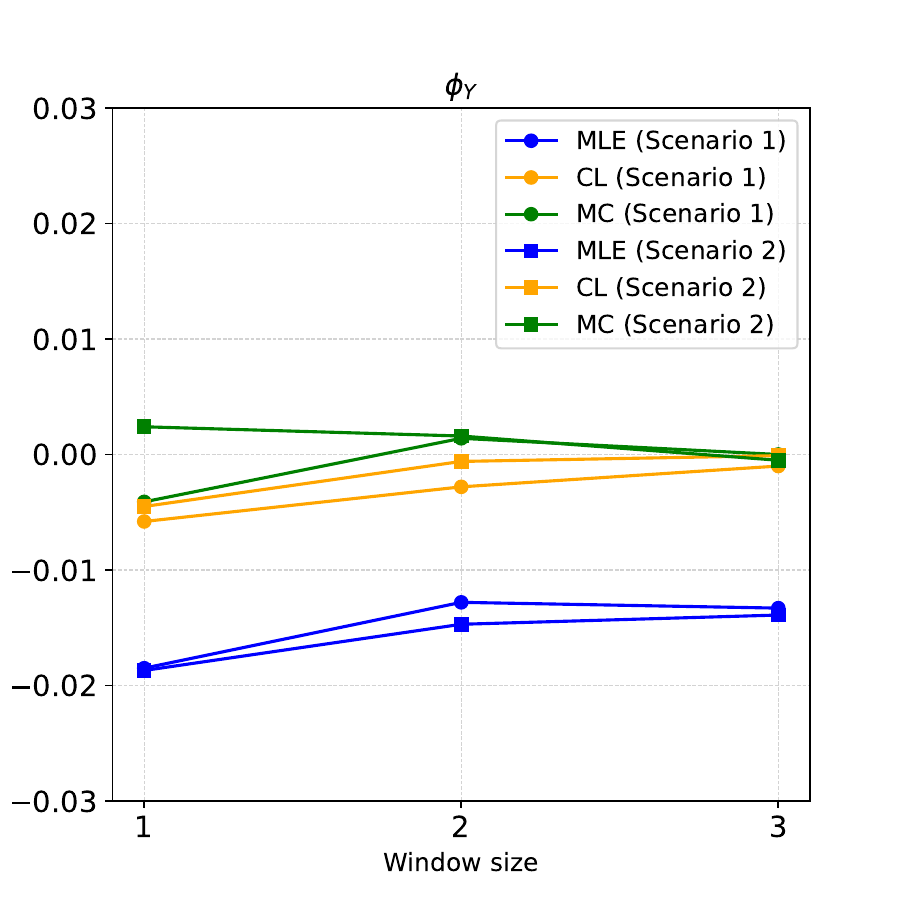}
\includegraphics[width=0.32\textwidth]{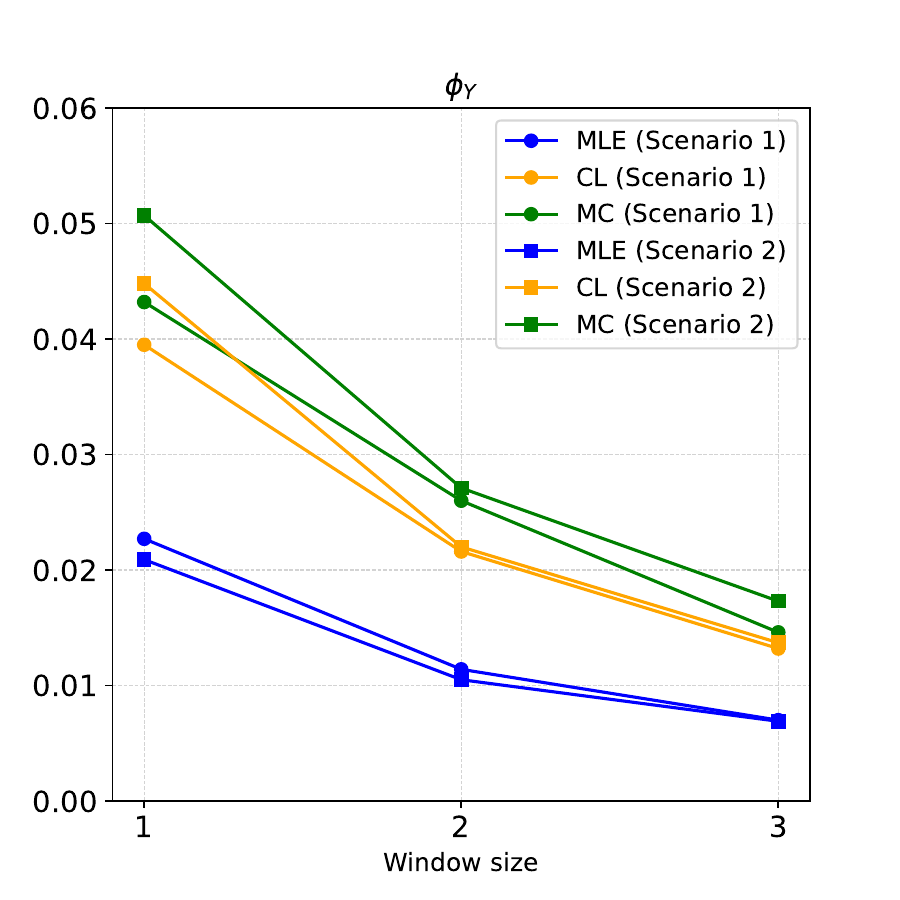}
\includegraphics[width=0.32\textwidth]{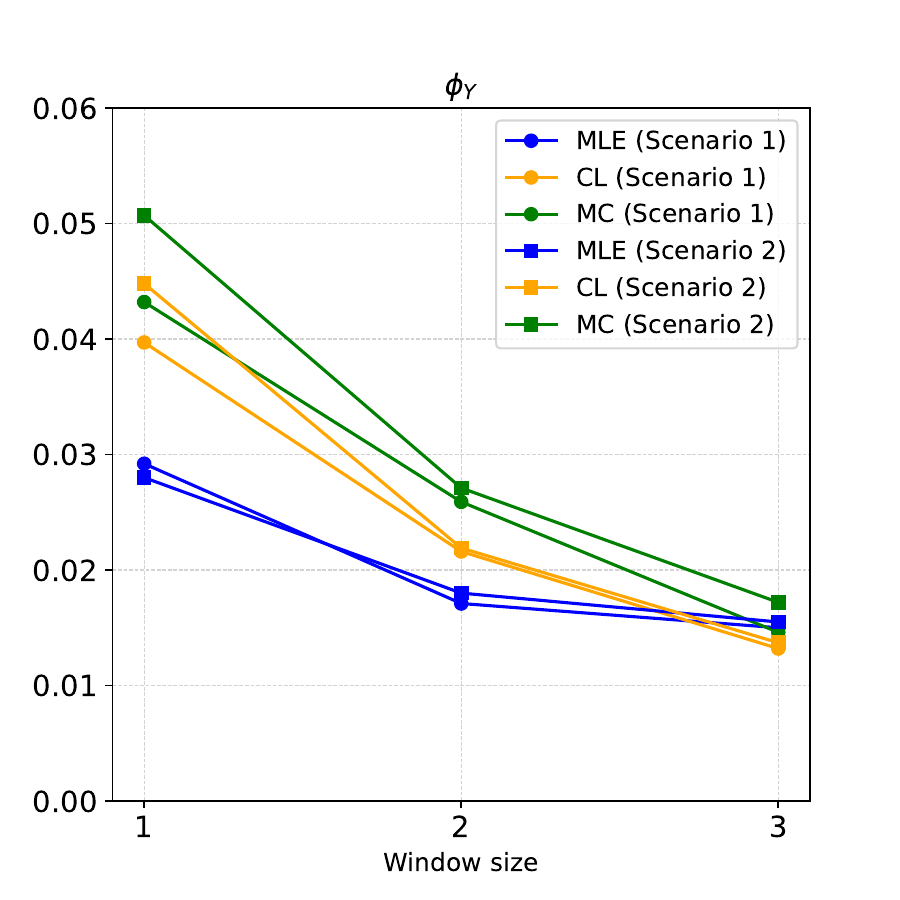}
\caption{Bias (left), SdErr (middle) and RMSE (right) of the estimated covariance parameters on simulations in Section~\ref{sec:5_2}.}\label{fig:2}
\end{figure}

\subsection{Asymptotic Analysis}
\label{sec:5_2}

We assess the asymptotic behaviour of our proposed estimators by expanding $S$ from $[0,1]\times [0,1]$ to $[0,3]\times [0,3]$, while maintaining the point density of 400 per unit square. In scenario 1, we fix $\nu_Y=0.5$ and $\phi_Y=0.1$. In scenario 2, we set $\sigma_{X}^2=1.2, \sigma_{XY}^2=1$ and $\sigma_{Y}^2=1$ with $\nu_X=1, \nu_{XY}=0.75, \nu_Y=0.5$ and $\phi_X=0.05, \phi_{XY}=0.07, \phi_Y=0.1$. We compare our results only with MLE, since TMB is too computationally intensive on $S=[0,3]\times [0,3]$. We test the bias-variance trade-off of MLE and our method using RMSE, with Bias and SdErr also plotted in Figures~\ref{fig:1}, \ref{fig:2} and \ref{fig:3}. The computation time is displayed in Figure~\ref{fig:4} (left). As $S$ expands, the standard errors decrease for both methods. For $\beta_0, \sigma_Y^2$ and $\phi_Y$, our method outperforms MLE in RMSE on $S=[0,2]\times [0,2]$ and $[0,3]\times [0,3]$, while being significantly faster. This advantage is expected to increase on larger data sets. For $\beta_1$ and $\sigma_e^2$, MLE achieves smaller RMSE due to lower estimation variances.

\begin{figure}[t]
\centering
\includegraphics[width=0.32\textwidth]{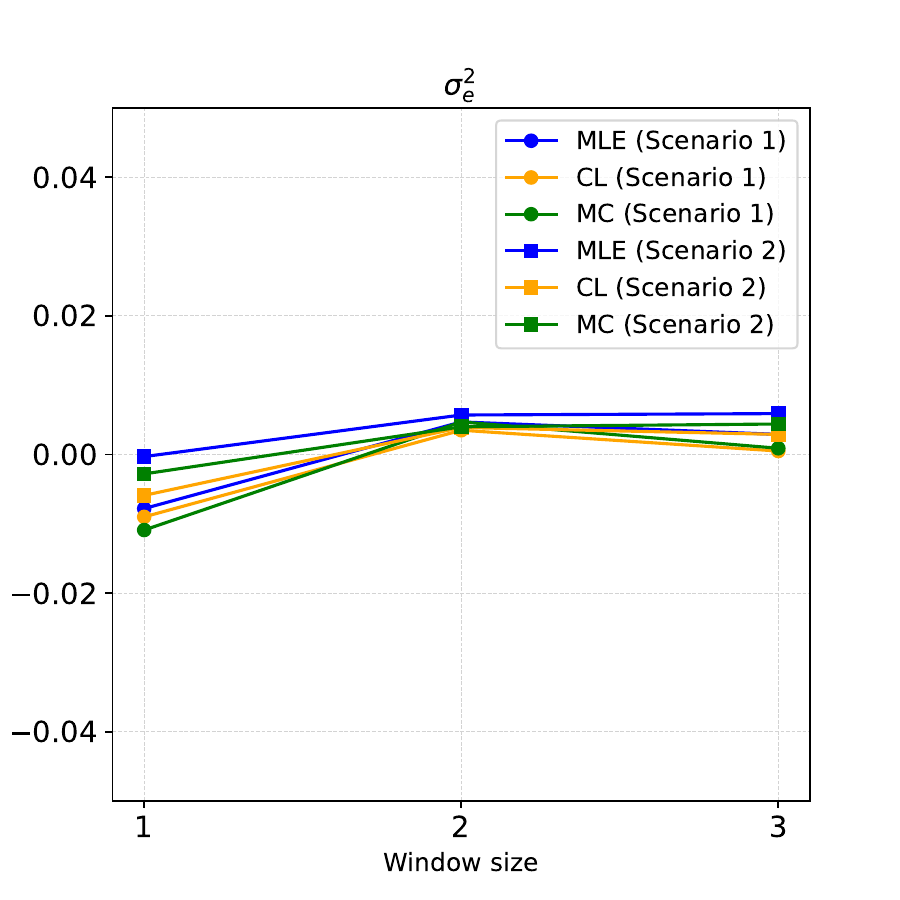}
\includegraphics[width=0.32\textwidth]{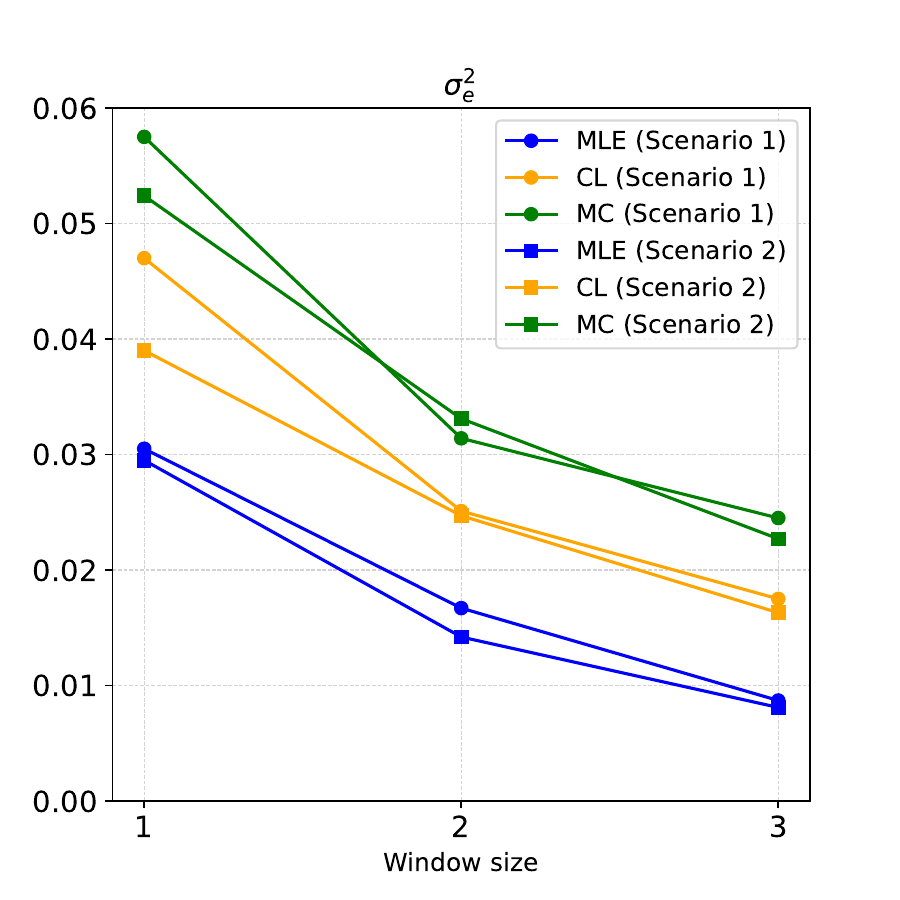}
\includegraphics[width=0.32\textwidth]{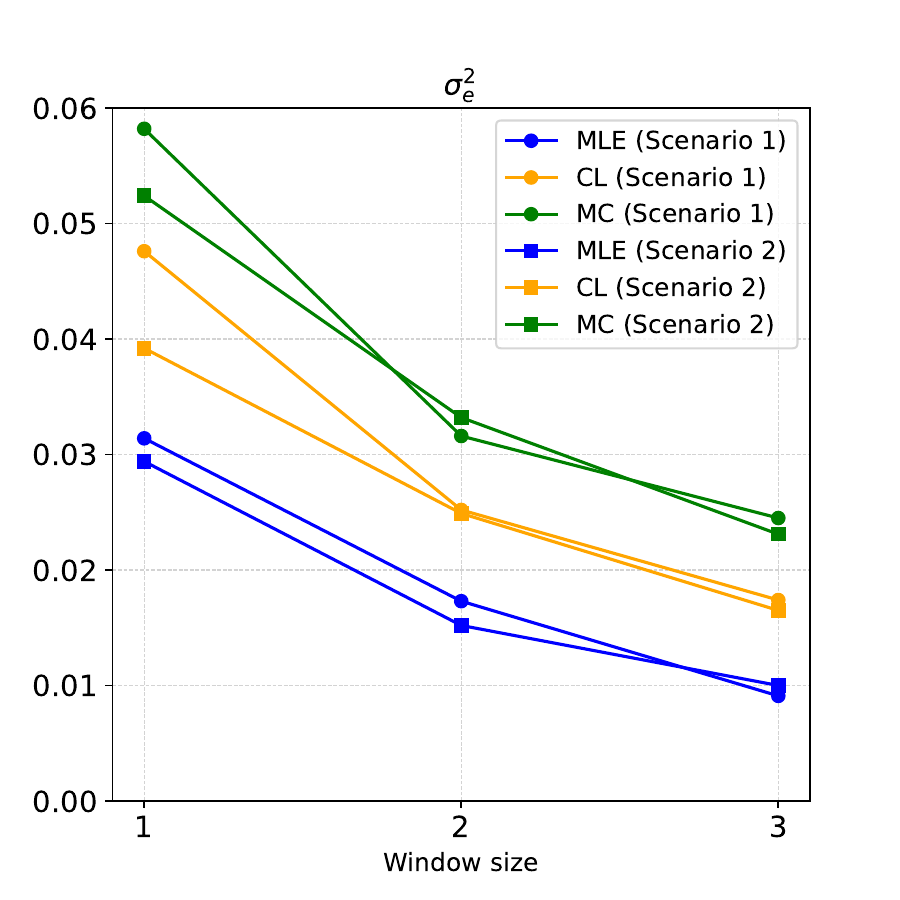}
\caption{Bias (left), SdErr (middle) and RMSE (right) of the estimated nugget effect on simulations in Section~\ref{sec:5_2}.}\label{fig:3}
\end{figure}

\begin{figure}[ht]
\centering
\includegraphics[width=0.32\textwidth]{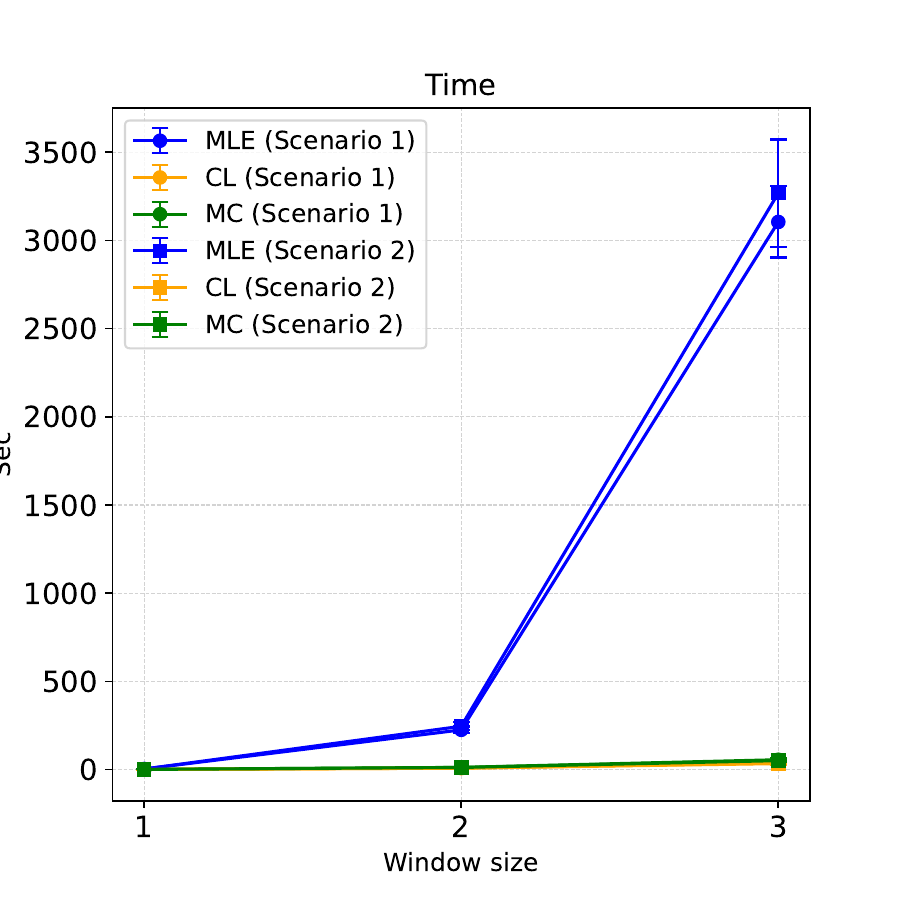}
\includegraphics[width=0.32\textwidth]{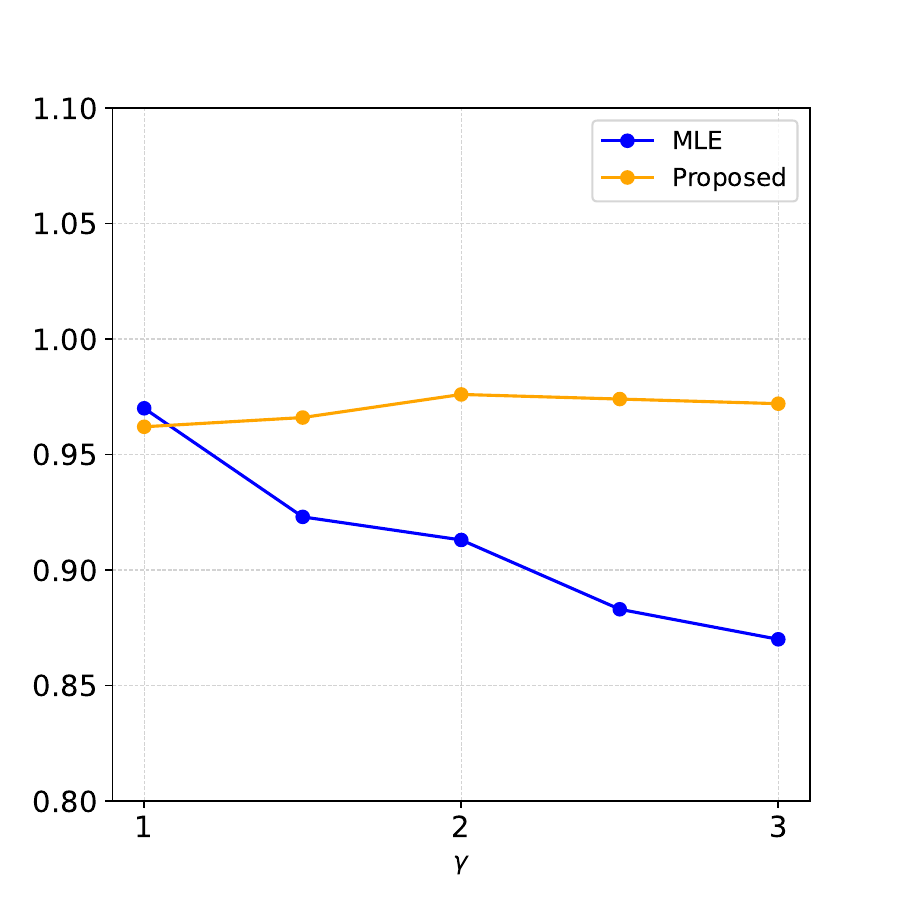}
\caption{Computation time (left) and confidence interval coverage for $\beta_1$ (right) on simulations in Section~\ref{sec:5_2}.}\label{fig:4}
\end{figure}

To examine the inference for the regression coefficients, we compute the 95\% confidence interval coverage of $\beta_1$ in scenario 1 on $S=[0,3]\times [0,3]$ by increasing the cross-correlation degree $\gamma$ from 1 to 3 with a step size of 0.5. The sill, covariance function and cross-covariance functions required for the asymptotic covariance matrix in Theorem~\ref{theorem:1} are estimated using the moment-based estimator (\ref{e:4}) and kernel-based estimators (\ref{e:5}) and (\ref{e:6}) with bandwidth selected via classical leave-one-out cross-validation. As shown in Figure~\ref{fig:4} (right), our method maintains coverage near 95\%, supporting the validity of Theorem~\ref{theorem:1} under preferential sampling. In contrast, the coverage obtained by MLE deviates far from 95\%. Note that its coverage of $\beta_0$ will be even worse because of estimation bias.

\section{An Application to Tropical Rainforest Data}
\label{sec:6}

To demonstrate the practical utility of the proposed method, we analyze data for `Trichilia tuberculata' trees, collected within a $200m \times 200m$ subregion of the 50 hectare forest dynamics plot on Barro Colorado Island in 1990. Tree diameters at breast height (DBH) are treated as marks. Figure~\ref{fig:5} (left) displays the tree locations and their DBH values. This data set was previously studied by \citet{myllymaki2009condition}, who modelled the tree locations as a stationary LGCP, with the latent Gaussian random field also governing the mark mean and covariance. Parameters were estimated using an empirical Bayesian approach: first fitting the point process and then analyzing the marks conditional on the fitted latent field. Their results suggest that trees in denser areas tend to have smaller diameters, as plotted in Figure~\ref{fig:5} (right).

\begin{figure}[t]
\centering
\includegraphics[width=0.42\textwidth]{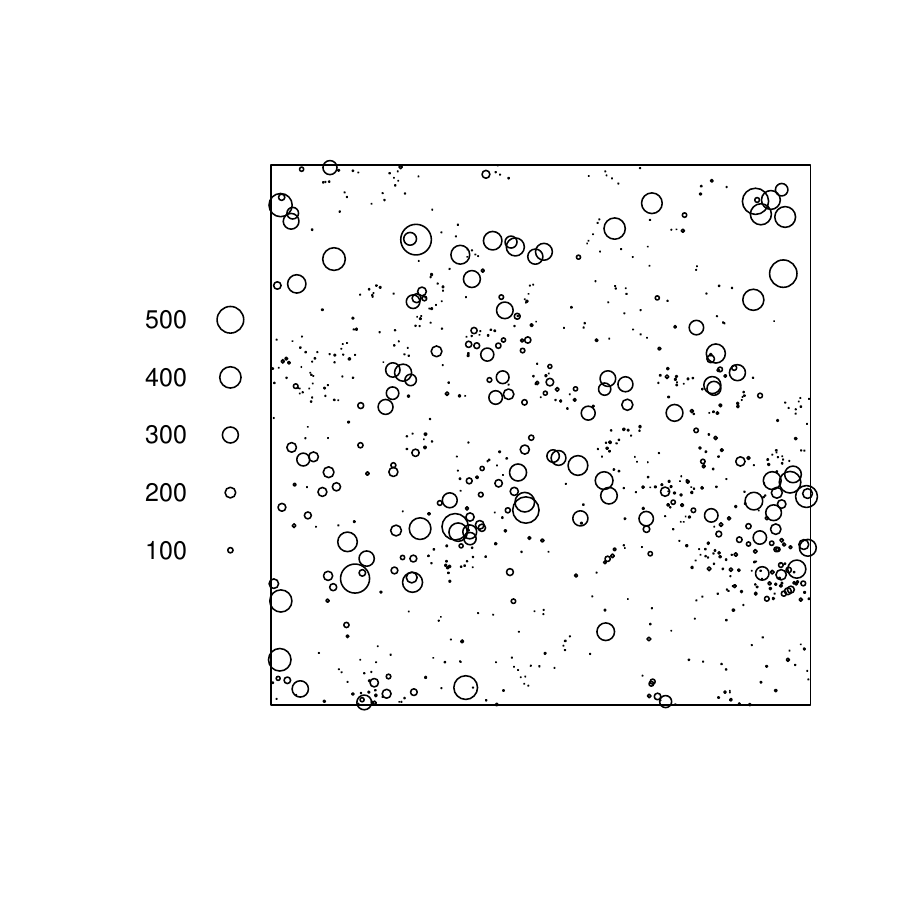}
\includegraphics[width=0.56\textwidth]{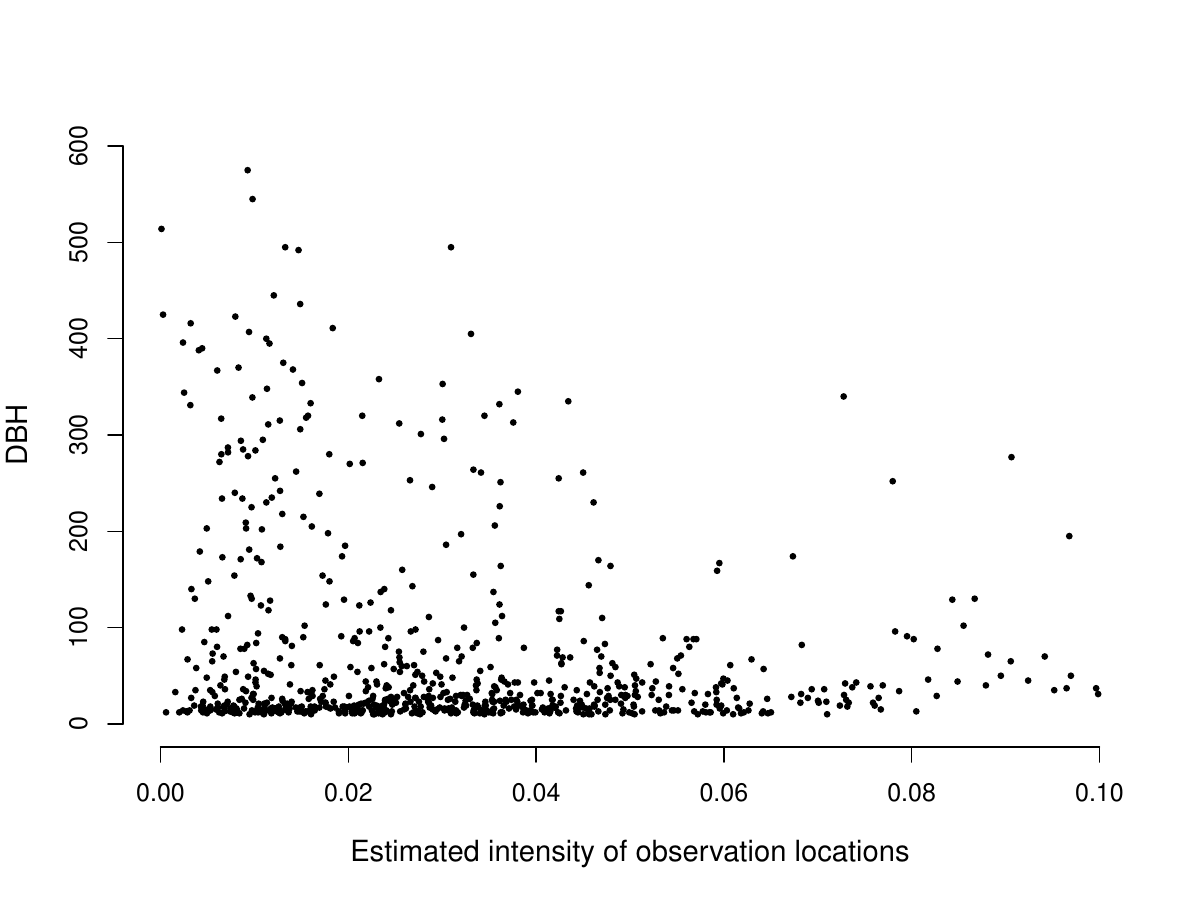}
\caption{Tree locations of `Trichilia tuberculata' (left), where circle size reflects the DBH, and scatter plot of DBH values against estimated local tree intensities (right).}\label{fig:5}
\end{figure}

We analyze the transformed marks $\log(\text{DBH}-9)$ and estimate its mean and covariance structures within a unified framework. Specifically, we include a spatial covariate, the square root of terrain slope, in the model (\ref{e:1}). Its spatial distribution is plotted in Figure~\ref{fig:6} (left). We assume the mark process has an exponential covariance function $C_Y(r)=\sigma_Y^2\exp(-r/\phi_Y)$ and choose $R=60$ in parameter estimation. The estimates for the regression coefficients, the covariance function and the nugget variance are reported in Table~\ref{tab:2}. For comparison, we apply MLE to the same data. Parameter estimates differ between MLE and our method, except for $\sigma_e^2$, as the former disregards preferential sampling. Our MC and CL estimators yield closely matching estimates for the covariance function. TMB meets computational issues on this data set, likely because its assumption $X(\bs)=\gamma Y(\bs)$ is violated. To investigate this, we estimate the covariance and cross-covariance functions using the estimators (\ref{e:4})--(\ref{e:6}) and plot them in Figure \ref{fig:6} (right). The dependence ranges of the two functions differ considerably, underscoring the importance of allowing flexible cross-correlation structures in geostatistical modelling. The negative values of the cross-covariance function again indicate that trees in areas of higher local intensity have smaller diameters, which is consistent with the findings in \citet{myllymaki2009condition}.

\begin{table}[t]
    \caption{Parameter estimates on tropical rainforest data.}
    \label{tab:2}
    \centering
    \scriptsize
    \begin{tabular}{|cccccc|}
    \hline
        Method & $\beta_0$ & $\beta_1$ & $\sigma_Y^2$ & $\phi_Y$ & $\sigma_e^2$\\[3pt]
    \hline
        MLE & 2.61 & 4.18 & 1.48 & 8.09 & 0.74\\
        CL & 2.84 & 3.32 & 1.31 & 11.28 & 0.78\\
        MC & 2.84 & 3.32 & 1.31 & 11.26 & 0.77\\
    \hline
    \end{tabular}
\end{table}

\begin{figure}[t]
\centering
\includegraphics[width=0.42\textwidth]{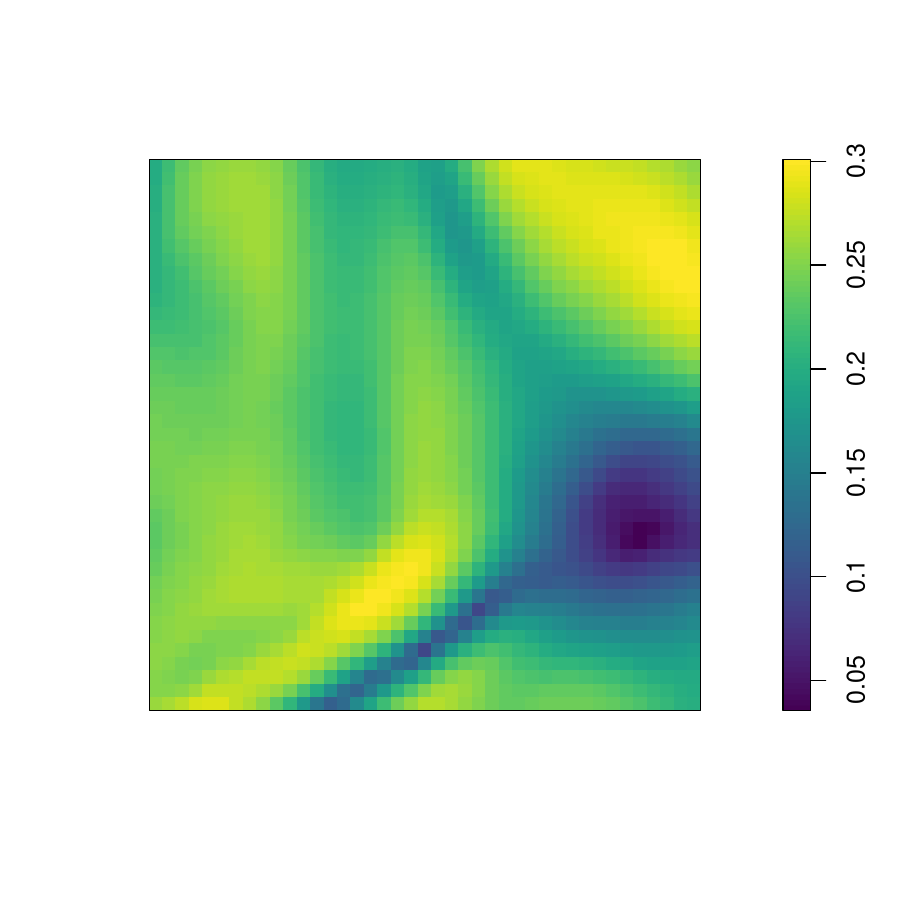}
\includegraphics[width=0.56\textwidth]{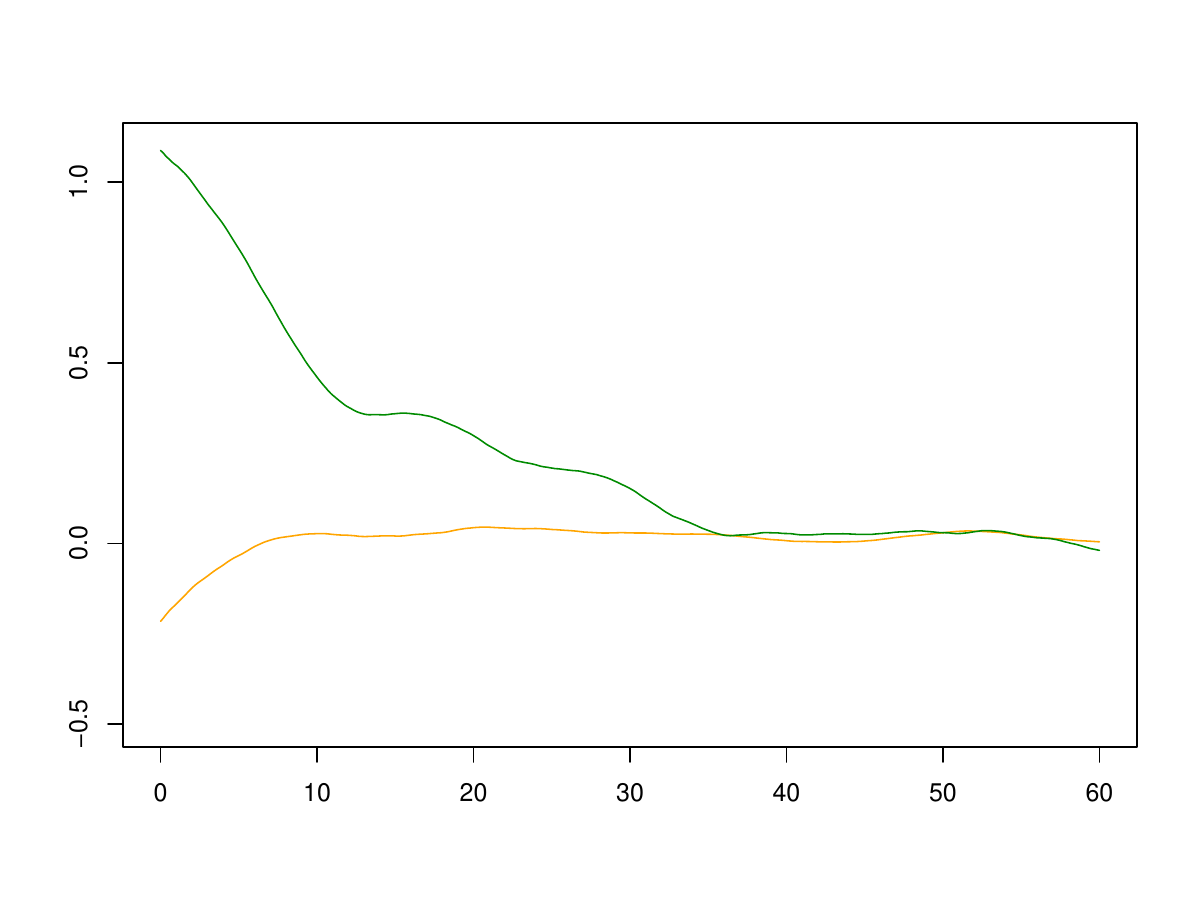}
\caption{Spatial distribution of the square root of terrain slope (left) and the estimated covariance (green) and cross-covariance (orange) functions (right).}\label{fig:6}
\end{figure}

\section{Conclusion}

This paper revisited the estimation of the mean and covariance structures in geostatistical models under preferential sampling. We relaxed the restrictive linear dependence assumption between the point process and the marks in the preferential sampling framework of \citet{diggle2010geostat}, allowing the cross-covariance function to take a general, isotropic form. Within this extended setting, we showed surprising findings that the least squares estimator for the regression coefficients and the kernel-based estimators for the spatial semi-variogram and cross-covariance remain consistent and unbiased, except for a bias in the intercept term of the regression coefficients. This bias can be corrected using the estimated cross-covariace at lag zero. Building on these results, we proposed unbiased estimators to infer the geostatistical model, without specifying a parametric sampling mechanism. Simulations under varying cross-correlation structures and an application to tropical rainforest data demonstrated that our method outperforms the likelihood-based approaches in estimation accuracy, computational efficiency and modelling flexibility.

For future work, the first direction would be to extend our proposed method to specialized geostatistical models under more complex preferential sampling mechanism, such that those with spatially varying regression coefficients or spatially varying sampling degrees. Second, it would be interesting to study the scenarios where the underlying point process moves beyond an LGCP, e.g.\ repulsive point processes that introduce inhibition among sampling locations.

\phantomsection\label{supplementary-material}
\bigskip

\begin{center}

{\large\bf SUPPLEMENTARY MATERIAL}

\end{center}

\begin{description}
\item[Lemma 1 and Its Proof:]
We will need the following lemma to prove Theorems~1--3.

\begin{lemma}
\label{lemma:1}
Suppose that $X$ and $Y$ are two Gaussian random variables with means $\mu_X,\mu_Y$. Then,
\begin{equation*}
\begin{split}
&\E{[Y\exp(X)]}=\left\{\Cov(X,Y)+\mu_Y \right\}\E[\exp(X)],\\
&\E{\left[Y^2\exp(X)\right]}=\left\{\Var(Y)+\left[\Cov(X,Y)+\mu_Y\right]^2\right\}\E{[\exp(X)]}.
\end{split}
\end{equation*}
\end{lemma}

\begin{proof}
    Let $X^\prime$ and $Y^\prime$ be zero-mean centred Gaussian random variables associated with $X$ and $Y$. By the Stein's Lemma,
    \begin{equation*}
    \begin{split}
        \E [Y\exp(X)]&=\E[(Y^\prime+\mu_Y)\exp(X^\prime+\mu_X)]\\
        &=\left\{\E [Y^\prime \exp(X^\prime)] +\mu_Y \E[\exp(X^\prime)]\right\}\exp(\mu_X)\\
        &=\left\{\Cov(X^\prime,Y^\prime)\E[\exp(X^\prime)]+\mu_Y \E[\exp(X^\prime)]\right\}\exp(\mu_X)\\
        &=\left\{\Cov(X,Y)+\mu_Y \right\}\E[\exp(X)],
    \end{split}
    \end{equation*}
    and
    \begin{equation*}
    \begin{split}
        \E [Y^2\exp(X)]&=\E[(Y^\prime+\mu_Y)^2\exp(X^\prime+\mu_X)]\\
        &=\left\{\E[(Y^\prime)^2\exp(X^\prime)]
        +2\mu_Y\E[Y^\prime\exp(X^\prime)]
        +\mu_Y^2\E[\exp(X^\prime)]\right\}\exp(\mu_X)\\
        &=\left\{[\Var(Y^\prime)+\Cov(X^\prime,Y^\prime)^2]+2\mu_Y\Cov(X^\prime,Y^\prime)+\mu_Y^2\right\}\\
        &\quad\quad\E[\exp(X^\prime)]\exp(\mu_X)\\
        &=\left\{\Var(Y)+[\Cov(X,Y)+\mu_Y]^2\right\}\E [\exp(X)].
    \end{split}
    \end{equation*}
\end{proof}

\item[Proof of Theorem 1:]

\begin{proof}
Under the defined asymptotic regime, the estimator (3) becomes
\begin{equation*}
    \tbbeta_n=\left[\sum_{\bs \in N\cap S_n}\bw(\bs)\bw(\bs)^\top\right]^{-1}\sum_{\bm{s}\in N\cap S_n}\bw(\bs)Z(\bs).
\end{equation*}
It minimizes $\sum_{\bs\in N\cap S_n} [Z(\bs)-\bw(\bs)^\top\bbeta]^2$, thus leads to the estimating equation:
\begin{equation*}
    \bre_n(\bbeta)=\sum_{\bs\in N\cap S_n} \left[Z(\bs)-\bw(\bs)^\top\bbeta\right]\bw(\bs)=\bm{0}.
\end{equation*}
Using the Taylor expansion, we obtain
\begin{equation*}
    |S_n|^{1/2}(\tbbeta_n-\bbeta^*_0)=\left[-\frac{\nabla \bre_n(\tilde{\bbeta}_n)}{|S_n|}\right]^{-1}\frac{\bre_n(\bbeta^*_0)}{|S_n|^{1/2}},
\end{equation*}
where $\nabla \bre_n(\bbeta)$ is the gradient of $\bre_n(\bbeta)$ with respect to $\bbeta$, and $\tilde{\bbeta}_n$ is a convex combination of $\tbbeta_n$ and $\bbeta^*_0$. 

First, by an application of the Campbell's theorem, we have 
\begin{equation*}
    \E\left[-\frac{\nabla \bre_n(\bbeta^*_0)}{|S_n|}\right]=\E\left[\frac{1}{|S_n|}\sum_{\bs\in N\cap S_n} \bw(\bs)\bw(\bs)^\top\right]=\frac{1}{|S_n|}\int_{S_n} \rho(\bs)\bw(\bs)\bw(\bs)^\top\rd\bs,
\end{equation*}
which, under conditions (C3)--(C4), is an $O(1)$. The variance of the $(i,j)$-th component in $\E[-\nabla \bre_n(\bbeta^*_0)/|S_n|]$ is
\begin{equation*}
\begin{split}
    &\frac{1}{|S_n|^2}\int_{S_n}\bw_{i}(\bs)^2\bw_{j}(\bs)^2\rho(\bs)\rd\bs\\
    &+\frac{1}{|S_n|^2}\int_{S_n}\int_{S_n}\bw_{i}(\bs)\bw_{j}(\bs)\bw_{i}(\bt)\bw_{j}(\bt)\left[\rho_2(\bs,\bt)-\rho(\bs)\rho(\bt)\right]\rd\bs\rd\bt,
\end{split}
\end{equation*}
which, under conditions (C3)--(C4), converges to zero as $n\to \infty$. 

Second, we analyze 
\begin{equation*}
    \frac{\bre_n(\bbeta^*_0)}{|S_n|^{1/2}}
    =\frac{1}{|S_n|^{1/2}}\left[Z(\bs)-\bw(\bs)^\top\bbeta^*_0\right]\bw(\bs).
\end{equation*}
Applying the Campbell's theorem, we have 
\begin{equation*}
\begin{split}
    \E\left[\frac{\bre_n(\bbeta^*_0)}{|S_n|^{1/2}}\right]
    =\frac{1}{|S_n|^{1/2}}\int_{S_n}&\E \{[Y(\bs)+e(\bs)-C_{XY}(0)]\exp[X(\bs)]\}\\
    &\lambda_0(\bs)\bw(\bs)\rd\bs=0. 
\end{split}
\tag{S1}
\end{equation*}
Moreover,
\begin{equation*}
\begin{split}
    \E\left[\frac{\bre_n(\bbeta^*_0)^2}{|S_n|}\right]
    % &=\E \left\{\frac{1}{|S_n|}\sum_{\bs\in N\cap S_n}[Y(\bs)+e(\bs)-C_{XY}(0)]\bw(\bs)\sum_{\bt\in N\cap S_n}[Y(\bt)+e(\bt)-C_{XY}(0)]\bw(\bt)^\top\right\}\\
    &=\frac{1}{|S_n|}\int_{S_n} \lambda_0(\bs)\E \left\{[Y(\bs)+e(\bs)-C_{XY}(0)]^2\exp[X(\bs)]\right\}\bw(\bs)\bw(\bs)^\top\rd\bs\\
    &+\frac{1}{|S_n|}\int_{S_n}\int_{S_n} \E\left\{[Y(\bs)+e(\bs)-C_{XY}(0)][Y(\bt)+e(\bt)-C_{XY}(0)]\right.\\
    &\quad\quad\quad\quad\quad\quad\ \left.\exp[X(\bs)+X(\bt)]\right\}
    \lambda_0(\bs)\lambda_0(\bt)\bw(\bs)\bw(\bt)^\top\rd\bs\rd\bt.
\end{split}
\end{equation*}
By Lemma 1 and the Isserlis' theorem,
\begin{equation*}
    \E \left\{[Y(\bs)+e(\bs)-C_{XY}(0)]^2\exp[X(\bs)]\right\}=\left(\sigma_Y^2+\sigma_e^2\right)\E \{\exp[X(\bs)]\}
\end{equation*}
and
\begin{equation*}
\begin{split}
    &\E\left\{[Y(\bs)+e(\bs)-C_{XY}(0)][Y(\bt)+e(\bt)-C_{XY}(0)]\exp[X(\bs)+X(\bt)]\right\}\\
    &=\left[C_Y(\|\bs-\bt\|)+C_{XY}(\|\bs-\bt\|)^2\right]\E \left\{\exp[X(\bs)+X(\bt)]\right\}.
\end{split}
\end{equation*}
Hence, we have
\begin{equation*}
\begin{split}
    \E\left[\frac{\bre_n(\bbeta^*_0)^2}{|S_n|}\right]
    =&\frac{1}{|S_n|}\int_{S_n} \rho(\bs)\left(\sigma_Y^2+\sigma_e^2\right)\bw(\bs)\bw(\bs)^\top\rd\bs\\
    +&\frac{1}{|S_n|}\int_{S_n}\int_{S_n} \rho_2(\bs,\bt)\left[C_Y(\|\bs-\bt\|)+C_{XY}(\|\bs-\bt\|)^2\right]\bw(\bs)\bw(\bt)^\top\rd\bs\rd\bt,
\end{split}
\tag{S2}
\end{equation*}
which is an $O(1)$ under conditions (C3)--(C6).

Next, we divide $S_n$ into $l_n$ disjoint sub-blocks of equal volume, up to negligible boundary corrections, which are asymptotically independent due to strong mixing. Denote a sub-block by $B_{n,i}$, thus $S_n=\cup_{i=1}^{l_n} B_{n,i}$. Write 
\begin{equation*}
    T_{n,i}=\sum_{\bs\in N\cap B_{n,i}} \left[Z(\bs)-\bw(\bs)^\top\bbeta\right]\bw(\bs).
\end{equation*}
According to (S1)--(S2), $\Var(\sum_{i=1}^{l_n} T_{n,i})=O(|S_n|)$. Following Lemma 1 in \citet{guan2007asymptotics},
\begin{equation*}
    |B_{n,i}|^2\E\left[\left(T_{n,i}/|B_{n,i}|\right)^4\right]<\infty
\end{equation*}
under conditions (C2)--(C4). Then, we have
\begin{equation*}
    \lim_{n\to\infty}\frac{\sum_{i=1}^{l_n}\E\left[\left(T_{n,i}\right)^4\right]}{\left[\Var\left(\sum_{i=1}^{l_n} T_{n,i}\right)\right]^2}
    =\lim_{n\to\infty}\frac{\sum_{i=1}^{l_n}|B_{n,i}|^2}{\left(\sum_{i=1}^{l_n}|B_{n,i}|\right)^2}=\lim_{n\to\infty}\frac{1}{l_n}=0.
\end{equation*}
By the Lyapunov's central limit theorem, $\bre_n(\bbeta^*_0)^2/|S_n|^{1/2}$ converges in distribution to a normally distributed vector with mean zero and covariance matrix as (S2).

The derivations above show that
\begin{equation*}
    -\frac{\nabla \bre_n(\bbeta^*_0)}{|S_n|}\xrightarrow{p}\frac{1}{|S_n|}\bA_n,\quad
    \frac{\bre_n(\bbeta^*_0)}{|S_n|^{1/2}}\xrightarrow{d}\mathrm{N(\bm{0},\bB_n+\bC_n)}.
\end{equation*}
By the Slutsky's theorem, the limiting covariance matrix of $\tbbeta_n$ is given by
\begin{equation*}
    \bSigma_n= |S_n|\bA_n^{-1}[\bB_n+\bC_n]\bA_n^{-1}.
\end{equation*}
The remainder of the proof follows arguments similar to those used in establishing the asymptotic normality of first-order estimating equations for point processes, e.g.\ \citet{schoenberg2005consistency}.
\end{proof}

\item[Proof of Theorem 2:]

\begin{proof}
    To establish the consistency of the estimator (5) to $V(r)$, we first show the convergence of the two sequences of random variables in the numerator and denominator, separately. For convenience, write
\begin{equation*}
\begin{split}
    A_{1,n}(r)&=\frac{1}{|S_n|}\mathop{\sum\sum}^{\neq}_{\bs,\bt\in N\cap S_n}\left\{\left[Z(\bs)-\bw(\bs)^\top\tbbeta_n\right]-\left[Z(\bt)-\bw(\bt)^\top\tbbeta_n\right]\right\}^2K_{h_n}(\|\bs-\bt\|-r),\\
    A_{2,n}(r)&=\frac{1}{|S_n|}\mathop{\sum\sum}^{\neq}_{\bs,\bt\in N\cap S_n}K_{h_n}(\|\bs-\bt\|-r).
\end{split}
\end{equation*}
Recall that $\tbbeta_n$ denotes the estimator (3) for $\bbeta^*_0$ on $S_n$. We define $\Delta\bbeta_n=\bbeta^*_0-\tbbeta_n$. Then, $A_{1,n}(r)$ can be decomposed as
\begin{equation*}
    A_{1,n}(r)=\mathcal{A}_n(r)+\mathcal{B}_n(r)+\mathcal{C}_n(r),
\end{equation*}
where 
\begin{equation*}
\begin{split}
    &\mathcal{A}_n(r)=\frac{1}{|S_n|}\mathop{\sum\sum}^{\neq}_{\bs,\bt\in N\cap S_n}\left[Y(\bs)+e(\bs)-Y(\bt)-e(\bt)\right]^2K_{h_n}(\|\bs-\bt\|-r),\\
    &\mathcal{B}_n(r)=\frac{2}{|S_n|}\mathop{\sum\sum}^{\neq}_{\bs,\bt\in N\cap S_n}\left[Y(\bs)+e(\bs)-Y(\bt)-e(\bt)\right]\left[\bw(\bs)^\top-\bw(\bt)^\top\right]\Delta\bbeta_n\\
    &\quad\quad\quad\quad\quad\quad\quad\quad\quad\ K_h(\|\bs-\bt\|-r),\\
    &\mathcal{C}_n(r)=\frac{1}{|S_n|}\mathop{\sum\sum}^{\neq}_{\bs,\bt\in N\cap S_n}\left\{\left[\bw(\bs)^\top-\bw(\bt)^\top\right]\Delta\bbeta_n\right\}^2K_h(\|\bs-\bt\|-r).
\end{split}
\end{equation*}
We now prove the convergence of the three sequences $\mathcal{A}_n(r),\mathcal{B}_n(r)$ and $\mathcal{C}_n(r)$. 

Consider $\mathcal{A}_n(r)$. By an application of the Campbell's theorem, we have
\begin{equation*}
\begin{split}
    \E[\mathcal{A}_n(r)]=&\frac{1}{|S_n|}\int_{S_n} \int_{S_n} \lambda_0(\bs)\lambda_0(\bt)\E\left\{\left[Y(\bs)+e(\bs)-Y(\bt)-e(\bt)\right]^2\exp[X(\bs)+X(\bt)]\right\}\\
    &\quad\quad\quad\quad\quad\ K_{h_n}(\|\bs-\bt\|-r)\rd\bs\rd\bt.
\end{split}
\end{equation*}
Since
\begin{equation*}
\begin{split}
    &\Var[Y(\bs)+e(\bs)-Y(\bt)-e(\bt)]=2\left[\sigma_Y^2+\sigma_e^2-C_Y(\|\bs-\bt\|)\right],\\
    &\Cov\left[Y(\bs)+e(\bs)-Y(\bt)-e(\bt),X(\bs)+X(\bt)\right]=0,
\end{split}
\end{equation*}
it follows from Lemma 1 that
\begin{equation*}
\begin{split}
    \E[\mathcal{A}_n(r)]=
    \frac{2}{|S_n|}\int_{S_n} \int_{S_n} \rho_2(\bs,\bt)\left[\sigma_Y^2+\sigma_e^2-C_Y(\|\bs-\bt\|)\right]K_{h_n}(\|\bs-\bt\|-r)\rd\bs\rd\bt.
\end{split}
\end{equation*}
With slight abuse of notation, set $a=\|\bs-\bt\|-r/h_n$ and write $\rho_2(\bs,\bt)=\rho_2\{\bs,\bs+(r+ah_n)[\cos(\psi),sin(\psi)]\}$ with $\psi\in[0,2\pi)$. Under conditions (C3)--(C4), there exists a constant $c_3>0$ such that 
\begin{equation*}
\begin{split}
    \left|\E[\mathcal{A}_n(r)]\right|\leq
    \frac{1}{|S_n|}\int_{S_n} c_3\rd\bs \int K(a)\rd a,
\end{split}
\end{equation*}
implying that $\E[\mathcal{A}_n(r)]= O(1)$. Moreover, as $h_n\to 0$ and under conditions (C3)--(C5), the following term dominates over the other higher-order terms in $\Var[\mathcal{A}_n(r)]$: 
\begin{equation*}
\begin{split}
    &\frac{2}{|S_n|^2}\int_{S_n} \int_{S_n} \lambda_0(\bs)\lambda_0(\bt)\E\left\{\left[Y(\bs)+e(\bs)-Y(\bt)-e(\bt)\right]^4\exp[X(\bs)+X(\bt)]\right\}\\
    &\quad\quad\quad\quad\quad K_{h_n}(\|\bs-\bt\|-r)^2\rd\bs\rd\bt.
\end{split}
\end{equation*}
By the Isserlis’ theorem and Lemma 1,  
\begin{equation*}
\begin{split}
    &\E\left\{[Y(\bs)+e(\bs)-Y(\bt)-e(\bt)]^4\exp[X(\bs)+X(\bt)]\right\}\\
    &=3\Var[Y(\bs)+e(\bs)-Y(\bt)-e(\bt)]\\
    &\quad\ \E\left\{[Y(\bs)+e(\bs)-Y(\bu)-e(\bt)]^2\exp[X(\bs)+X(\bt)]\right\}
\end{split}
\end{equation*}
and
\begin{equation*}
\begin{split}
    &\E\left\{[Y(\bs)+e(\bs)-Y(\bt)-e(\bt)]^2\exp[X(\bs)+X(\bt)]\right\}\\
    &=2\left[\sigma_Y^2+\sigma_e^2-C_Y(\|\bs-\bt\|)\right]\E\left\{\exp[X(\bs)+X(\bt)]\right\}.
\end{split}
\end{equation*}
Then, we have 
\begin{equation*}
\begin{split}
    &\E\left\{[Y(\bs)+e(\bs)-Y(\bt)-e(\bt)]^4\exp[X(\bs)+X(\bt)]\right\}\\
    &=12\left[\sigma_Y^2+\sigma_e^2-C_Y(\|\bs-\bt\|)\right]^2\E\left\{\exp[X(\bs)+X(\bt)]\right\}.
\end{split}
\end{equation*}
Under conditions (C3)--(C4), there exists a constant $c_4>0$ such that
\begin{equation*}
\begin{split}
    &\frac{2}{|S_n|^2}\int_{S_n} \int_{S_n} \lambda_0(\bs)\lambda_0(\bt)\E\left\{\left[Y(\bs)+e(\bs)-Y(\bt)-e(\bt)\right]^4\exp[X(\bs)+X(\bt)]\right\}\\
    &\quad\quad\quad\quad\quad\ \ \ K_{h_n}(\|\bs-\bt\|-r)^2\rd\bs\rd\bt\\
    &\leq \frac{2}{|S_n|^2}\int_{S_n} c_4 \rd\bs \int \frac{1}{h_n}K(a)^2\rd a=O\left(\frac{1}{|S_n|h_n}\right).
\end{split}
\end{equation*}
As $|S_n|h_n\to \infty$, this term vanishes, implying that
\begin{equation*}
     \mathcal{A}_n(r)\xrightarrow{p} \frac{2}{|S_n|}\int_{S_n} \int_{S_n} \rho_2(\bs,\bt)\left[\sigma_Y^2+\sigma_e^2-C_Y(\|\bs-\bt\|)\right]K_{h_n}(\|\bs-\bt\|-r)\rd\bs\rd\bt
\tag{S3}
\end{equation*}
with convergence rate $O[(|S_n|h_n)^{-1/2}]$.

Consider $\mathcal{B}_n(r)$. By Theorem 1 and under condition (C3), applying the Campbell’s theorem gives
\begin{equation*}
\begin{split}
    \E[\mathcal{B}_n(r)]=&\frac{2}{|S_n|}\int_{S_n} \int_{S_n} \lambda_0(\bs)\lambda_0(\bt) \E\left\{\left[Y(\bs)+e(\bs)-Y(\bt)-e(\bt)\right]\exp[X(\bs)+X(\bt)]\right\}\\
    &\quad\quad\quad\quad\quad\ K_h(\|\bs-\bt\|-r)O\left(\frac{1}{|S_n|^{1/2}}\right)\rd\bs\rd\bt.
\end{split}
\end{equation*}
By Lemma1,
\begin{equation*}
    \E\left\{\left[Y(\bs)+e(\bs)-Y(\bt)-e(\bt)\right]\exp[X(\bs)+X(\bt)]\right\}=0.
\end{equation*}
Hence, we have $\E[\mathcal{B}_n(r)]=0$.
Moreover, as $h_n\to 0$ and under condition (C3)--(C5), the following term dominates over the other higher-order terms in $\Var[\mathcal{B}_n(r)]$: 
\begin{equation*}
\begin{split}
    &\frac{2}{|S_n|^2}\int_{S_n} \int_{S_n} \lambda_0(\bs)\lambda_0(\bt)\E\left\{\left[Y(\bs)+e(\bs)-Y(\bt)-e(\bt)\right]^2\exp[X(\bs)+X(\bt)]\right\}\\
    &\quad\quad\quad\quad\quad K_{h_n}(\|\bs-\bt\|-r)^2 O\left(\frac{1}{|S_n|}\right) \rd\bs\rd\bt.
\end{split}
\end{equation*}
As derived above,
\begin{equation*}
\begin{split}
    &\E\left\{[Y(\bs)+e(\bs)-Y(\bt)-e(\bt)]^2\exp[X(\bs)+X(\bt)]\right\}\\
    &=2\left[\sigma_Y^2+\sigma_e^2-C_Y(\|\bs-\bt\|)\right]\E\left\{\exp[X(\bs)+X(\bt)]\right\}.
\end{split}
\end{equation*}
Under conditions (C3)--(C4), there exists a constant $c_5>0$ such that 
\begin{equation*}
\begin{split}
    &\frac{2}{|S_n|^2}\int_{S_n} \int_{S_n} \lambda_0(\bs)\lambda_0(\bt)\E\left\{\left[Y(\bs)+e(\bs)-Y(\bt)-e(\bt)\right]^2\exp[X(\bs)+X(\bt)]\right\}\\
    &\quad\quad\quad\quad\quad K_{h_n}(\|\bs-\bt\|-r)^2 O\left(\frac{1}{|S_n|}\right) \rd\bs\rd\bt\\
    &\leq \frac{2}{|S_n|^2}\int_{S_n} \frac{c_5}{|S_n|} \rd\bs \int \frac{1}{h_n}K(a)^2\rd a=O\left(\frac{1}{|S_n|^2h_n}\right).
\end{split}
\end{equation*}
As $|S_n|\to \infty$ and $|S_n|h_n\to \infty$, this term vanishes, implying that
\begin{equation*}
    \mathcal{B}_n(r)\xrightarrow{p} 0
\tag{S4}
\end{equation*}
with convergence rate $O[(|S_n|^2h_n)^{-1/2}]$. Note that this rate is faster than that of $\mathcal{A}_n(r)$.

Similarly, for $\mathcal{C}_n(r)$, it can be shown that 
$
    \mathcal{C}_n(r)\xrightarrow{p} 0,
$
with convergence rate $O[(|S_n|^3h_n)^{-1/2}]$. By the continuous mapping theorem and recalling (S3)--(S4), we obtain 
\begin{equation*}
    A_{1,n}(r)\xrightarrow{p} \frac{2}{|S_n|}\int_{S_n} \int_{S_n} \rho_2(\bs,\bt)\left[\sigma_Y^2+\sigma_e^2-C_Y(\|\bs-\bt\|)\right]K_{h_n}(\|\bs-\bt\|-r)\rd\bs\rd\bt
\tag{S5}
\end{equation*}
with convergence rate $O[(|S_n|h_n)^{-1/2}]$.

Second, by an application of the Campbell’s theorem, we have
\begin{equation*}
    \E[A_{2,n}(r)]=\frac{1}{|S_n|}\int_{S_n}\int_{S_n} \rho_2(\bs,\bt)K_{h_n}(\|\bs-\bt\|-r)\rd\bs\rd\bt.
\end{equation*}
Under conditions (C3)--(C4), there exists a constant $c_6>0$ such that 
\begin{equation*}
\begin{split}
    \left|\E[A_{2,n}(r)]\right|\leq
    \frac{1}{|S_n|}\int_{S_n} c_6\rd\bs \int K(a)\rd a,
\end{split}
\end{equation*}
implying that $\E[\mathcal{A}_n(r)]= O(1)$. Moreover, as $h_n\to 0$ and under condition (C3)--(C5), the following term dominates over the other higher-order terms in $\Var[A_{2,n}(r)]$: 
\begin{equation*}
    \frac{2}{|S_n|^2}\int_{S_n} \int_{S_n} \rho_2(\bs,\bt)K_{h_n}(\|\bs-\bt\|-r)^2\rd\bs\rd\bt,
\end{equation*}
and can be bounded by
\begin{equation*}
    \frac{2}{|S_n|^2}\int_{S_n} c_7 \rd\bs \int \frac{1}{h_n}K(a)^2\rd a=O\left(\frac{1}{|S_n|h_n}\right)
\end{equation*}
under conditions (C3)--(C4), where $c_7$ is a constant $>0$. As $|S_n|h_n\to \infty$, the term above converges to zero, implying that 
\begin{equation*}
    \E[A_{2,n}(r)]\xrightarrow{p} \frac{1}{|S_n|}\int_{S_n}\int_{S_n} \rho_2(\bs,\bt)K_{h_n}(\|\bs-\bt\|-r)\rd\bs\rd\bt
\end{equation*}
with convergence rate $O[(|S_n|h_n)^{-1/2}]$. By the continuous mapping theorem and recalling (S5), we obtain 
\begin{equation*}
    \frac{A_{1,n}(r)}{2A_{2,n}(r)} \xrightarrow{p} \frac{\int_{S_n} \int_{S_n} \rho_2(\bs,\bt)\left[\sigma_Y^2+\sigma_e^2-C_Y(\|\bs-\bt\|)\right]K_{h_n}(\|\bs-\bt\|-r)\rd\bs\rd\bt}{\int_{S_n}\int_{S_n} \rho_2(\bs,\bt)K_{h_n}(\|\bs-\bt\|-r)\rd\bs\rd\bt}
\tag{S6}
\end{equation*}
with convergence rate $O[(|S_n|h_n)^{-1/2}]$.

Next, we quantify the difference 
\begin{equation*}
    \mathcal{D}_n(r)=\frac{\int_{S_n} \int_{S_n} \rho_2(\bs,\bt)\left[\sigma_Y^2+\sigma_e^2-C_Y(\|\bs-\bt\|)\right]K_{h_n}(\|\bs-\bt\|-r)\rd\bs\rd\bt}{\int_{S_n}\int_{S_n} \rho_2(\bs,\bt)K_{h_n}(\|\bs-\bt\|-r)\rd\bs\rd\bt}-V(r).
\end{equation*}
Replacing $\bt$ with $\bs+(r+ah_n)[\cos(\psi),\sin(\psi)]$, we rewrite the first term in $\mathcal{D}_n(r)$ as
\begin{equation*}
    \frac{\int_{S_n}\int_0^{2\pi}\int \rho_2\{\bs,\bs+(r+ah_n)[\cos(\psi),\sin(\psi)]\}\left[\sigma_Y^2+\sigma_e^2-C_Y(r+ah_n)\right]K(a)\rd a\rd\psi\rd\bs}{\int_{S_n}\int_0^{2\pi}\int \rho_2\{\bs,\bs+(r+ah_n)[\cos(\psi),\sin(\psi)]\}K(a)\rd a\rd\psi\rd\bs}
\end{equation*}
Since $V_Y(r)$ is smooth in a neighbourhood of $r$, as $h_n\to 0$, the covariance function $C_Y(r)$ admits a first-order Taylor expansion:
\begin{equation*}
    C_Y(r+a_nh_n)=C_Y(r)+ah_nC_Y^\prime(r)+O(h_n^2).
\end{equation*}
where $C_Y^\prime(r)$ denotes the derivative of $C_Y(r)$. Under conditions (C3)--(C4), there exists a constant $c_8>0$ such that
\begin{equation*}
\begin{split}
    \mathcal{D}_n(r)\leq \frac{\int_{S_n}\int_0^{2\pi}c_8\rd\psi\rd\bs \int \left[|a|h_nC^\prime(r)+O(h_n^2)\right]K(a)\rd a}{\int_{S_n}\int_0^{2\pi}\int \rho_2\{\bs,\bs+(r+ah_n)[\cos(\psi),\sin(\psi)]\}K(a)\rd a\rd\psi\rd\bs}=O(h_n).
\end{split}
\end{equation*}
Then, recalling (S6), we have 
\begin{equation*}
    \frac{A_{1,n}(r)}{2A_{2,n}(r)}\xrightarrow{p} V_Y(r)
\end{equation*}
with convergence rate $O[h_n+(|S_n|h_n)^{-1/2}]$.

To establish consistency of the estimator (6) to $C_{XY}(r)$, we need only to show consistency of the sequence of random variables
\begin{equation*}
    \mathcal{E}_n(r)=\frac{1}{|S_n|}\mathop{\sum\sum}^{\neq}_{\bs,\bt\in N\cap S_n}\left[Y(\bs)+e(\bs)-C_{XY}(0)\right]K_{h_n}(\|\bs-\bt\|-r).
\end{equation*}
By an application of the Campbell's theorem, we have
\begin{equation*}
\begin{split}
    \E[\mathcal{E}_n(r)]&=\int_{S_n} \int_{S_n} \lambda_0(\bs)\lambda_0(\bt)\E\left\{[Y(\bs)+e(\bs)-C_{XY}(0)]\exp[X(\bs)+X(\bt)]\right\}\\
    &\quad\quad\quad\ \ K_{h_n}(\|\bs-\bt\|-r)\rd\bs\rd\bt.
\end{split}
\end{equation*}
By Lemma 1, 
\begin{equation*}
    \E\{[Y(\bs)+e(\bs)-C_{XY}(0)]\exp[X(\bs)+X(\bt)]\}=C_{XY}(\|\bs-\bt\|)\E\{\exp[X(\bs)+X(\bt)]\}.
\end{equation*}
Then, we have
\begin{equation*}
    \E[\mathcal{E}_n(r)]=\frac{1}{|S_n|}\int_{S_n} \int_{S_n} \rho_2(\bs,\bt)C_{XY}(\|\bs-\bt\|)K_{h_n}(\|\bs-\bt\|-r)\rd\bs\rd\bt.
\end{equation*}
The remainder of the proof proceeds with the same procedure as for $V_Y(r)$.

Furthermore, for the sill estimator, note that
\begin{equation*}
    \E\left\{\frac{1}{|S_n|}\sum_{\bs\in N\cap S_n}\left[Z(\bs)-\bw(\bs)^\top\tbbeta_n\right]^2\right\}=\frac{1}{|S_n|}\int_{S_n} \left(\sigma_Y^2+\sigma_e^2\right)\rho(\bs)\rd\bs
\end{equation*}
and 
\begin{equation*}
    \E\left(\frac{|N|}{|S_n|}\right)=\frac{1}{|S_n|}\int_{S_n} \rho(\bs)\rd\bs,
\end{equation*}
its consistency therefore follows by a similar argument.
\end{proof}

\item[Proof of Theorem 3:]

Note that the two estimating equations (\ref{e:9}) and (\ref{e:11}) follow a general form:
\begin{equation*}
    \bU(\btheta)=\mathop{\sum\sum}^{\neq}_{\bs,\bt\in N}\rw(\bs,\bt)\bzeta^{(1)}(\|\bs-\bt\|;\btheta)\left\{\left[\hat{Z}(\bs)-\hat{Z}(\bt)\right]^2-2\zeta(\|\bs-\bt\|;\btheta)\right\}=\bm 0.
\end{equation*}

Write $\brw(\bs,\bt)=\rw(\bs,\bt)\bzeta^{(1)}(\|\bs-\bt\|;\btheta)$ and
\begin{equation*}
    \bU_n^*(\btheta)=\mathop{\sum\sum}^{\neq}_{\bs,\bt\in N\cap S_n}\brw(\bs,\bt)\left\{\left[Z^*(\bs)-Z^*(\bt)\right]^2-2\zeta(\|\bs-\bt\|;\btheta)\right\}=\bm 0.
\end{equation*}
By Theorem 1 and following similar arguments in the proof of Theorem 2, we need only to show that $\E[\bU_n^*(\btheta_0)/|S_n|^2]=\bm 0$ and $\Var[\bU_n^*(\btheta_0)/|S_n|^2]$ converges to zero when $n\to\infty$.

First, we have
\begin{equation*}
\begin{split}
    \E\left[\frac{\bU_n^*(\btheta_0)}{|S_n|^2}\right]
    =\frac{1}{|S_n|^2} \int_{S_n}\int_{S_n}&\lambda_0(\bs)\lambda_0(\bt)\brw(\bs,\bt)\E\left(\left\{[Y(\bs)+e(\bs)-Y(\bt)-e(\bt)]^2\right.\right.\\
    &-\left.\left.2\left[\sigma_Y^2+\sigma_e^2-C_Y(\bs-\bt)\right]\right\}\exp[X(\bs)+X(\bt)]\right)\rd\bs\rd\bt.
\end{split}
\end{equation*}
By Lemma 1,
\begin{equation*}
\begin{split}
    &\E\left\{[Y(\bs)+e(\bs)-Y(\bt)-e(\bt)]^2\exp[X(\bs)+X(\bt)]\right\}\\
    &=2\left[\sigma_Y^2+\sigma_e^2-C_Y(\bs-\bt)\right]\E\left\{\exp[X(\bs)+X(\bt)]\right\}.
\end{split}
\end{equation*}
Hence, under conditions (C3)--(C4) and (C8)--(C9), $\E[\bU_n^*(\btheta_0)/|S_n|^2]=\bm 0$.

Second, we analyze
\begin{equation*}
\begin{split}
    \Var\left[\frac{\bU_n^*(\btheta_0)}{|S_n|^2}\right]=&\frac{2}{|S_n|^4}\int_{S_n}\int_{S_n}\lambda_0(\bs)\lambda_0(\bt)\brw(\bs,\bt)\brw(\bs,\bt)^{\top}\\
    &\quad\E\left(\left\{\left[Z^*(\bs)-Z^*(\bt)\right]^2-2\zeta(\|\bs-\bt\|;\btheta_0)\right\}^2\right.\\
    &\quad\quad\left.\exp[X(\bs)+X(\bt)]\right)\rd\bs\rd\bt\\
    +&\frac{4}{|S_n|^4}\int_{S_n}\int_{S_n}\int_{S_n}
    \lambda_0(\bs)\lambda_0(\bt)\lambda_0(\bu)\brw(\bs,\bt)\brw(\bs,\bu)^{\top}\\
    &\quad\E\left(\exp[X(\bs)+X(\bt)+X(\bu)]\right.\\
    &\quad\quad\left\{\left[Z^*(\bs)-Z^*(\bt)\right]^2-2\zeta(\|\bs-\bt\|;\btheta_0)\right\}\\
    &\quad\quad\left\{\left.\left[Z^*(\bs)-Z^*(\bu)\right]^2-2\zeta(\|\bs-\bu\|;\btheta_0)\right\}\right)\rd\bs\rd\bt\rd\bu\\
    +&\frac{1}{|S_n|^4}\int_{S_n}\int_{S_n}\int_{S_n}\int_{S_n}
    \lambda_0(\bs)\lambda_0(\bt)\lambda_0(\bu)\lambda_0(\bv)
    \brw(\bs,\bt)\brw(\bu,\bv)^{\top}\\
    &\quad\E\left(\exp[X(\bs)+X(\bt)+X(\bu)+X(\bv)]\right.\\
    &\quad\quad\left\{\left[Z^*(\bs)-Z^*(\bt)\right]^2-2\zeta(\|\bs-\bt\|;\btheta_0)\right\}\\
    &\quad\quad\left\{\left.\left[Z^*(\bu)-Z^*(\bv)\right]^2-2\zeta(\|\bu-\bv\|;\btheta_0)\right\}\right)
    \rd\bs\rd\bt\rd\bu\rd\bv.
\end{split}
\end{equation*}
Consider the fourth-order term. By the Isserlis's theorem,
\begin{equation*}
\begin{split}
    &\E\left\{\left[Z^*(\bs)-Z^*(\bt)\right]^2\left[Z^*(\bu)-Z^*(\bv)\right]^2\exp[X(\bs)+X(\bt)+X(\bu)+X(\bv)]\right\}\\
    &=\Var[Y(\bs)+e(\bs)-Y(\bt)-e(\bt)]\\
    &\quad\ \E\left\{[Y(\bu)+e(\bu)-Y(\bv)-e(\bv)]^2\exp[X(\bs)+X(\bt)+X(\bu)+X(\bv)]\right\}\\
    &+2\Cov\left[Y(\bs)+e(\bs)-Y(\bt)-e(\bt),Y(\bu)+e(\bu)-Y(\bv)-e(\bv)\right]\\
    &\quad\ \E\left\{[Y(\bs)+e(\bs)-Y(\bt)-e(\bt)][Y(\bu)+e(\bu)-Y(\bv)-e(\bv)]\right.\\
    &\quad\quad\left. \exp[X(\bs)+X(\bt)+X(\bu)+X(\bv)]\right\}\\
    &+\Cov\left[Y(\bs)+e(\bs)-Y(\bt)-e(\bt),X(\bs)+X(\bt)+X(\bu)+X(\bv)\right]\\
    &\quad\ \E\left\{[Y(\bs)+e(\bs)-Y(\bt)-e(\bt)][Y(\bu)+e(\bu)-Y(\bv)-e(\bv)]^2\right.\\
    &\quad\quad\left. \exp[X(\bs)+X(\bt)+X(\bu)+X(\bv)]\right\}.
\end{split}
\end{equation*}
By Lemma 1, 
\begin{equation*}
\begin{split}
    &\E\left\{[Y(\bu)+e(\bu)-Y(\bv)-e(\bv)]^2\exp[X(\bs)+X(\bt)+X(\bu)+X(\bv)]\right\}\\
    &=\left\{\Cov[Y(\bu)+e(\bu)-Y(\bv)-e(\bv),X(\bs)+X(\bt)+X(\bu)+X(\bv)]^2\right.\\
    &\quad+\left.\Var[Y(\bu)+e(\bu)-Y(\bv)-e(\bv)]\right\}\E\left\{\exp[X(\bs)+X(\bt)+X(\bu)+X(\bv)]\right\}.
\end{split}   
\end{equation*}
Moreover,
\begin{equation*}
\begin{split}
    &\E\left\{[Y(\bs)+e(\bs)-Y(\bt)-e(\bt)][Y(\bu)+e(\bu)-Y(\bv)-e(\bv)]\right.\\
    &\quad\left. \exp[X(\bs)+X(\bt)+X(\bu)+X(\bv)]\right\}\\
    &=\left\{\Cov[Y(\bs)+e(\bs)-Y(\bt)-e(\bt),Y(\bu)+e(\bu)-Y(\bv)-e(\bv)]\right.\\
    &\quad+\Cov[Y(\bs)+e(\bs)-Y(\bt)-e(\bt),X(\bs)+X(\bt)+X(\bu)+X(\bv)]\\
    &\quad\quad\left.\Cov[Y(\bu)+e(\bu)-Y(\bv)-e(\bv),X(\bs)+X(\bt)+X(\bu)+X(\bv)]\right\}\\
    &\quad\ \E\left\{\exp[X(\bs)+X(\bt)+X(\bu)+X(\bv)]\right\},
\end{split}
\end{equation*}
and
\begin{equation*}
\begin{split}
    &\E\left\{[Y(\bs)+e(\bs)-Y(\bt)-e(\bt)][Y(\bu)+e(\bu)-Y(\bv)-e(\bv)]^2\right.\\
    &\quad\left. \exp[X(\bs)+X(\bt)+X(\bu)+X(\bv)]\right\}\\
    &=\left\{2\Cov[Y(\bs)+e(\bs)-Y(\bt)-e(\bt),Y(\bu)+e(\bu)-Y(\bv)-e(\bv)]\right.\\
    &\quad\quad\ \Cov[Y(\bu)+e(\bu)-Y(\bv)-e(\bv),X(\bs)+X(\bt)+X(\bu)+X(\bv)]\\
    &\quad\ +\Cov[Y(\bs)+e(\bs)-Y(\bt)-e(\bt),X(\bs)+X(\bt)+X(\bu)+X(\bv)]\\
    &\quad\quad\ \Var[Y(\bu)+e(\bu)-Y(\bv)-e(\bv)]\\
    &\quad\ +\Cov[Y(\bs)+e(\bs)-Y(\bt)-e(\bt),X(\bs)+X(\bt)+X(\bu)+X(\bv)]\\
    &\quad\quad\ \left.\Cov[Y(\bu)+e(\bu)-Y(\bv)-e(\bv),X(\bs)+X(\bt)+X(\bu)+X(\bv)]^2\right\}\\
    &\quad\ \ \E\left\{\exp[X(\bs)+X(\bt)+X(\bu)+X(\bv)]\right\}
\end{split}
\end{equation*}
Then, we have
\begin{equation*}
\begin{split}
    &\E\left\{\left[Z^*(\bs)-Z^*(\bt)\right]^2\left[Z^*(\bu)-Z^*(\bv)\right]^2\exp[X(\bs)+X(\bt)+X(\bu)+X(\bv)]\right\}\\
    &=4\left[\sigma_Y^2+\sigma_e^2-C_Y(\|\bs-\bt\|)\right]\left[\sigma_Y^2+\sigma_e^2-C_Y(\|\bu-\bv\|)\right]\\
    &+2\left[\sigma_Y^2+\sigma_e^2-C_Y(\|\bs-\bt\|)\right]\\
    &\quad[C_{XY}(\bs-\bu)-C_{XY}(\bs-\bv)+C_{XY}(\bt-\bu)-C_{XY}(\bt-\bv)]^2\\
    &+2\left[C_Y(\bs-\bu)-C_Y(\bt-\bu)-C_Y(\bs-\bv)+C_Y(\bt-\bv)\right]^2\\
    &+4\left[C_Y(\bs-\bu)-C_Y(\bt-\bu)-C_Y(\bs-\bv)+C_Y(\bt-\bv)\right]\\
    &\quad\left[C_{XY}(\bs-\bu)-C_{XY}(\bt-\bu)+C_{XY}(\bs-\bv)-C_{XY}(\bt-\bv)\right]\\
    &\quad\left[C_{XY}(\bs-\bu)-C_{XY}(\bs-\bv)+C_{XY}(\bt-\bu)-C_{XY}(\bt-\bv)\right]\\
    &+2\left[\sigma_Y^2+\sigma_e^2-C_Y(\|\bu-\bv\|)\right]\\
    &\quad\left[C_{XY}(\bs-\bu)-C_{XY}(\bt-\bu)+C_{XY}(\bs-\bv)-C_{XY}(\bt-\bv)\right]^2\\
    &+\left[C_{XY}(\bs-\bu)-C_{XY}(\bt-\bu)+C_{XY}(\bs-\bv)-C_{XY}(\bt-\bv)\right]^2\\
    &\quad\left[C_{XY}(\bs-\bu)-C_{XY}(\bs-\bv)+C_{XY}(\bt-\bu)-C_{XY}(\bt-\bv)\right]^2.
\end{split}
\end{equation*}
Hence, 
\begin{equation*}
\begin{split}
    &\int_{S_n}\int_{S_n}\int_{S_n}\int_{S_n}
    \lambda_0(\bs)\lambda_0(\bt)\lambda_0(\bu)\lambda_0(\bv)
    \brw(\bs,\bt)\brw(\bu,\bv)^{\top}\\
    &\quad\E\left(\exp[X(\bs)+X(\bt)+X(\bu)+X(\bv)]\left\{\left[Z^*(\bs)-Z^*(\bt)\right]^2-2\zeta(\|\bs-\bt\|;\btheta_0)\right\}\right.\\
    &\quad\quad\left\{\left.\left[Z^*(\bu)-Z^*(\bv)\right]^2-2\zeta(\|\bu-\bv\|;\btheta_0)\right\}\right)
    \rd\bs\rd\bt\rd\bu\rd\bv\\
    &=\int_{S_n}\int_{S_n}\int_{S_n}\int_{S_n}\rho_4(\bs,\bt,\bu,\bv)
    \brw(\bs,\bt)\brw(\bu,\bv)^{\top}\\
    &\quad\ \ \left\{2\left[C_Y(\bs-\bu)-C_Y(\bt-\bu)-C_Y(\bs-\bv)+C_Y(\bt-\bv)\right]^2\right.\\
    &\quad\ \ +4\left[C_Y(\bs-\bu)-C_Y(\bt-\bu)-C_Y(\bs-\bv)+C_Y(\bt-\bv)\right]\\
    &\quad\quad\ \ \left[C_{XY}(\bs-\bu)-C_{XY}(\bt-\bu)+C_{XY}(\bs-\bv)-C_{XY}(\bt-\bv)\right]\\
    &\quad\quad\ \ \left[C_{XY}(\bs-\bu)-C_{XY}(\bs-\bv)+C_{XY}(\bt-\bu)-C_{XY}(\bt-\bv)\right]\\
    &\quad\ \ +\left[C_{XY}(\bs-\bu)-C_{XY}(\bt-\bu)+C_{XY}(\bs-\bv)-C_{XY}(\bt-\bv)\right]^2\\
    &\quad\quad\ \ \left.[C_{XY}(\bs-\bu)-C_{XY}(\bs-\bv)+C_{XY}(\bt-\bu)-C_{XY}(\bt-\bv)]^2\right\}\rd\bs\rd\bt\rd\bu\rd\bv,
\end{split}
\tag{S7}
\end{equation*}
which, under conditions (C3)--(C6) and (C8)--(C9), is an $O(|S_n|^3)$. Here, $\rho_4(\bs,\bt,\bu,\bv)$ is the fourth-order factorial density function of $N$. Similarly, we have
\begin{equation*}
\begin{split}
    &\E\left\{\left[Z^*(\bs)-Z^*(\bt)\right]^2\left[Z^*(\bs)-Z^*(\bu)\right]^2\exp[X(\bs)+X(\bt)+X(\bu)]\right\}\\
    &=2\left[\sigma_Y^2+\sigma_e^2-C_Y(\|\bs-\bt\|)\right]\\
    &\quad\left\{2\left[\sigma_Y^2+\sigma_e^2-C_Y(\|\bs-\bu\|)\right]+[C_{XY}(\bs-\bt)-C_{XY}(\bt-\bu)]^2\right\}\\
    &+2\left[\sigma_Y^2-C_Y(\bs-\bt)-C_Y(\bs-\bu)+C_Y(\bt-\bu)\right]^2\\
    &+4\left[\sigma_Y^2-C_Y(\bs-\bt)-C_Y(\bs-\bu)+C_Y(\bt-\bu)\right]\left[C_{XY}(\bs-\bu)-C_{XY}(\bt-\bu)\right]\\
    &\quad\left[C_{XY}(\bs-\bt)-C_{XY}(\bt-\bu)\right]\\
    &+\left\{2\left[\sigma_Y^2+\sigma_e^2-C_Y(\|\bs-\bu\|)\right]+[C_{XY}(\bs-\bt)-C_{XY}(\bt-\bu)]^2\right\}\\
    &\quad\left[C_{XY}(\bs-\bu)-C_{XY}(\bt-\bu)\right]^2.
\end{split}
\end{equation*}
Hence,
\begin{equation*}
\begin{split}
    &\int_{S_n}\int_{S_n}\int_{S_n}
    \lambda_0(\bs)\lambda_0(\bt)\lambda_0(\bu)\brw(\bs,\bt)\brw(\bs,\bu)^{\top}\\
    &\quad\E\left(\exp[X(\bs)+X(\bt)+X(\bu)]\left\{\left[Z^*(\bs)-Z^*(\bt)\right]^2-2\zeta(\|\bs-\bt\|;\btheta_0)\right\}\right.\\
    &\quad\quad\left\{\left.\left[Z^*(\bs)-Z^*(\bu)\right]^2-2\zeta(\|\bs-\bu\|;\btheta_0)\right\}\right)\rd\bs\rd\bt\rd\bu\\
    &=\int_{S_n}\int_{S_n}\int_{S_n}\rho_3(\bs,\bt,\bu)\brw(\bs,\bt)\brw(\bs,\bu)^{\top}\\
    &\quad\ \ \left\{2\left[\sigma_Y^2-C_Y(\bs-\bt)-C_Y(\bs-\bu)+C_Y(\bt-\bu)\right]^2\right.\\
    &\quad\ \ +4\left[\sigma_Y^2-C_Y(\bs-\bt)-C_Y(\bs-\bu)+C_Y(\bt-\bu)\right]\\
    &\quad\quad\ \ \left[C_{XY}(\bs-\bu)-C_{XY}(\bt-\bu)\right]\left[C_{XY}(\bs-\bt)-C_{XY}(\bt-\bu)\right]\\
    &\quad\ \ +\left.\left[C_{XY}(\bs-\bu)-C_{XY}(\bt-\bu)\right]^2[C_{XY}(\bs-\bt)-C_{XY}(\bt-\bu)]^2\right\}\rd\bs\rd\bt\rd\bu,
\end{split}  
\tag{S8}
\end{equation*}
which, under conditions (C3)--(C6) and (C8)--(C9), is also an $O(|S_n|^3)$. Here, $\rho_3(\bs,\bt,\bu)$ is the third-order factorial density function of $N$. Furthermore,
\begin{equation*}
\begin{split}
    &\int_{S_n}\int_{S_n}\lambda_0(\bs)\lambda_0(\bt)\brw(\bs,\bt)\brw(\bs,\bt)^{\top}\\
    &\quad\E\left(\left\{\left[Z^*(\bs)-Z^*(\bt)\right]^2-2\zeta(\|\bs-\bt\|;\btheta_0)\right\}^2\exp[X(\bs)+X(\bt)]\right)\rd\bs\rd\bt\\
    &=8\int_{S_n}\int_{S_n}\rho_2(\bs,\bt)\brw(\bs,\bt)\brw(\bs,\bt)^{\top}
    \left[\sigma_Y^2+\sigma_e^2-C_Y(\bs-\bt)\right]\rd\bs\rd\bt,
\end{split}
\tag{S9}
\end{equation*}
which, under conditions (C3)--(C4) and (C8)--(C9), is an $O(|S_n|^2)$. Collecting the results of (S7)--(S9), we have $\Var[\bU_n^*(\btheta_0)/|S_n|^2]=O(|S_n|^{-1})$, which converges to zero as $n\to\infty$.

Finally, to determine the convergence rate, under Theorem 1, we consider the Taylor expansion:
\begin{equation*}
    |S_n|^{1/2}(\tbtheta_n-\btheta_0)=\left[-\frac{\nabla \bU_n^*(\tilde{\btheta}_n)}{|S_n|^2}\right]^{-1}\frac{\bU_n^*(\btheta_0)}{|S_n|^{3/2}},
\end{equation*}
where $\tbtheta_n$ represents $\tbtheta_{n,MC}$ and $\tbtheta_{n,CL}$, $\nabla U_n^*(\btheta)$ is the gradient of $U_n^*(\btheta)$ with respect to $\btheta$, and $\tilde{\btheta}_n$ is a convex combination of $\tbtheta_n$ and $\btheta_0$. Similar to the derivations above, both $\E[-\nabla \bU_n^*(\btheta_0)/|S_n|^2]$ and $\Var[U_n^*(\btheta_0)/|S_n|^{3/2}]$ converge to constants as $n\to\infty$. Hence, the rate is $O(|S_n|^{-1/2})$.

\end{description}

\end{document}